\documentclass[11pt,english]{article}

\usepackage[table]{xcolor}
\usepackage{lineno}
\usepackage{amsmath,amssymb,amsthm}
\usepackage{multicol}
\usepackage[margin=1in]{geometry}
\usepackage{graphicx,color}
\usepackage{enumitem}
\usepackage{fullpage}
\usepackage[noblocks]{authblk}
\usepackage{tcolorbox}
\usepackage{babel}
\usepackage{wrapfig}
\usepackage{MnSymbol}
\usepackage{mdframed}
\usepackage{mathtools}
\usepackage{floatrow}
\usepackage{multirow}
\usepackage{tablefootnote}
\usepackage[leftcaption]{sidecap}
\usepackage{hhline}

\usepackage[ruled,linesnumbered,vlined]{algorithm2e}
	\usepackage{tablefootnote}
	\usepackage{thm-restate}
\usepackage{bbding}
\usepackage{pifont}
\usepackage{nicefrac}
\usepackage{undertilde}

\definecolor{Darkblue}{rgb}{0,0,0.4}
\definecolor{Brown}{cmyk}{0,0.61,1.,0.60}
\definecolor{Purple}{cmyk}{0.45,0.86,0,0}
\definecolor{Darkgreen}{rgb}{0.133,0.543,0.133}

\usepackage[colorlinks,linkcolor=Darkblue,filecolor=blue,citecolor=blue,urlcolor=Darkblue,pagebackref]{hyperref}
\usepackage[nameinlink]{cleveref}

\usepackage[colorinlistoftodos,prependcaption,textsize=tiny]{todonotes}

\newcommand{\atodoin}[1]{\todo[linecolor=red,backgroundcolor=green!25,bordercolor=red,inline]{\textbf{AF:~}#1}}

\newcommand{\atodoAfter}[1]{}
\newcommand{\aidea}[1]{}

\newif\ifdraft 
\draftfalse

\newcommand{\namedref}[2]{\hyperref[#2]{#1~\ref*{#2}}}
\newcommand{\propref}[1]{\hyperref[#1]{property~(\ref*{#1})}}

\newlength{\Oldarrayrulewidth}

\newtheorem{theorem}{Theorem}
\newtheorem{lemma}{Lemma}
\newtheorem{definition}{Definition}
\newtheorem{claim}{Claim}

\newtheorem{corollary}{Corollary}

\newtheorem{remark}{Remark}

\newcommand{\poly}{\mathrm{poly}}
\newcommand{\polylog}{\mathrm{polylog}}

\newcommand{\R}{\mathbb{R}}
\newcommand{\Z}{\mathbb{Z}}
\newcommand{\N}{\mathbb{N}}

\newcommand{\opt}{\mathrm{opt}}

\newcommand{\diam}{\mathrm{diam}}

\newcommand{\UTSP}{\textsf{UTSP}\xspace}
\newcommand{\TSP}{\textsf{TSP}\xspace}
\newcommand{\UST}{\textsf{UST}\xspace}
\newcommand{\OPT}{\text{OPT}}
\newcommand{\PRDO}{\text{PRDO}\xspace}

\newcommand{\CURW}{\textsf{CURW}\xspace}
\newcommand{\SPCS}{\textsf{SPCS}\xspace}

\newcommand{\frechet}{Fr\'echet}

 \DeclareMathOperator{\cost}{cost}

 \newcommand{\dom}{\ensuremath{\mathrm{dom}}}
\def\cA{\ensuremath{\mathcal{A}}}  
  
\def\cC{\ensuremath{\mathcal{C}}}

\def\cF{\ensuremath{\mathcal{F}}}

\def\cK{\ensuremath{\mathcal{K}}}

\def\cP{\ensuremath{\mathcal{P}}}

\def\cS{\ensuremath{\mathcal{S}}}
\def\cT{\ensuremath{\mathcal{T}}}

\definecolor{forestgreen}{rgb}{0.13, 0.55, 0.13}

\SetCommentSty{mycommfont}

\def\eps{\varepsilon}

\DeclareMathAlphabet{\mathpzc}{OT1}{pzc}{m}{it}
\newcommand{\etal}{{\em et al. \xspace}}

\newcommand{\rt}{\mbox{\rm rt}}

\newlength{\dhatheight}

\newcommand {\ignore} [1] {}

\newcommand{\initOneLiners}{%
	\setlength{\itemsep}{0.2pt}
	\setlength{\parsep }{0.2pt}
	\setlength{\topsep }{0.2pt}
}

\definecolor{BrickRed}{rgb}{.72,0,0}

\title{On Sparse Covers of Minor Free Graphs, Low Dimensional Metric Embeddings, and other applications \footnote{The original version of this paper also contained a general reduction from sparse covers to padded decomposition. This reduction was moved to a different manuscript which will be made public soon.}}

\author{Arnold Filtser \thanks{Email: \texttt{arnold.filtser@biu.ac.il}. This research was supported by the ISRAEL SCIENCE FOUNDATION (grant No. 1042/22).}\\Bar-Ilan University}
\date{}
\begin{document}
\maketitle
\begin{abstract}
Given a metric space $(X,d_X)$, a $(\beta,s,\Delta)$-sparse cover is a collection of clusters $\mathcal{C}\subseteq P(X)$ with diameter at most $\Delta$, such that for every point $x\in X$, the ball $B_X(x,\frac\Delta\beta)$ is fully contained in some cluster $C\in \mathcal{C}$, and $x$ belongs to at most $s$ clusters in $\mathcal{C}$.
Our main contribution is to show that the shortest path metric of every $K_r$-minor free graphs admits $(O(r),O(r^2),\Delta)$-sparse cover, and for every $\epsilon>0$, $(4+\epsilon,O(\frac1\epsilon)^r,\Delta)$-sparse cover (for arbitrary $\Delta>0$). We then use this sparse cover to show that every $K_r$-minor free graph embeds into $\ell_\infty^{\tilde{O}(\frac1\epsilon)^{r+1}\cdot\log n}$ with distortion $3+\eps$ (resp. into $\ell_\infty^{\tilde{O}(r^2)\cdot\log n}$ with distortion $O(r)$). 
Further, among other applications, this sparse cover immediately implies an algorithm for the oblivious buy-at-bulk problem in fixed minor free graphs with the tight approximation factor $O(\log n)$ (previously nothing beyond general graphs was known).
%	Finally, 
%	we provide applications of these sparse covers into padded decompositions, sparse partitions, universal TSP / Steiner tree, oblivious buy at bulk, name independent routing, and path reporting distance oracles.
\end{abstract}

%\newpage
%\setcounter{secnumdepth}{5}
%\setcounter{tocdepth}{3} \tableofcontents

%\vfill
\newpage
%\begin{multicols}{2}
%	{\small 
\setcounter{secnumdepth}{5}
		\setcounter{tocdepth}{2} \tableofcontents
%	}
%\end{multicols}

%\setcounter{tocdepth}{2} 
%\tableofcontents
    \newpage
    \pagenumbering{arabic}
%    \linenumbers

\section{Introduction}

Given a metric space $(X,d_X)$ a $(\beta,s,\Delta)$-sparse cover is a collection of clusters $\cC\subseteq P(X)$ all with diameter at most $\Delta$, such that every point belongs to at most $s$ clusters, and every ball $B_X(x,\frac{\Delta}{\beta})$ is fully contained in some cluster $C\in\cC$. % (see \Cref{def:SparseCover}).
Sparse covers are very useful for  algorithmic design, and in particular for divide and concur. 
Since their introduction by Awerbuch and Peleg \cite{AP90}, sparse covers found numerous applications. A partial list of which includes: 
compact routing schemes \cite{PU89,Peleg00,TZ01b,AGMNT08,AGM05,AGMW10,BLT14},
distant-dependent distributed directories \cite{AP91,Peleg93,Peleg00,BLT14},
network synchronizers \cite{AW04,AP90b,Lynch96,Peleg00,BLT14},
distributed deadlock prevention \cite{AKP94},
construction of spanners and ultrametric covers \cite{HIS13,FN22,LS23,HMO23,FL22,FGN24},
metric embeddings \cite{Rao99,KLMN04},
universal TSP and Stiner tree constructions \cite{JLNRS05,BDRRS12,Fil20scattering,BCFHHR24}, and Oblivious buy-at-bulk \cite{SBI11}.

We will study sparse covers in weighted graphs, where we distinguish between two different types of diameter. The \emph{weak} diameter of a cluster $A\subseteq V$ in a weighted graph $G=(V,E,w)$ is the maximum pairwise distance $\max_{u,v\in A}d_G(u,v)$ w.r.t. the original shortest path distance, while the \emph{strong} diameter is the maximum pairwise distance $\max_{u,v\in A}d_{G[A]}(u,v)$ in the induced subgraph. We continue with a formal definition:
\begin{definition}[Sparse Cover]\label{def:SparseCover}
	Given a weighted graph $G=(V,E,w)$, a collection of clusters $\mathcal{C} = \{C_1,..., C_t\}$ is called a weak/strong $(\beta,s,\Delta)$ sparse cover if the following conditions hold.
	\begin{enumerate}
		\item Bounded diameter: The weak/strong diameter of every cluster $C_i\in\mathcal{C}$ is bounded by $\Delta$.\label{condition:RadiusBlowUp}
		\item Padding: For each $v\in V$, there exists a cluster $C_i\in\mathcal{C}$ such that $B_G(v,\frac\Delta\beta)\subseteq C_i$.
		\item Sparsity: For each $v\in V$, there are at most $s$ clusters in $\mathcal{C}$ containing $v$.		
	\end{enumerate}	
	If the clusters $\cC$ can be partitioned into $s$ partitions $\cP_1,\dots,\cP_s$ s.t. $\cC=\cup_{i=1}^s\cP_i$, then  $\{\cP_1,\dots,\cP_s\}$ is called a weak/strong $(\beta,s,\Delta)$ sparse partition cover.  
	We say that a graph $G$ admits a weak/strong $(\beta,s)$ sparse (partition) cover scheme, if for every parameter $\Delta>0$ it admits a weak/strong $(\beta,s,\Delta)$ sparse (partition) cover that can be constructed in expected polynomial time. Sparse partition cover scheme is abbreviated \SPCS.
\end{definition}
The notion of sparse cover scheme is the more common in the literature. However, \SPCS provides additional structure that is crucial for different applications. \footnote{In this paper we require the partition property of \SPCS for metric embeddings into $\ell_\infty$, and for the oblivious buy-at-bulk. This property is also crucial in the construction of ultrametric covers \cite{FL22,Fil23,FGN24}.}
Obtaining strong diameter guarantee is also considerably more challenging, however it is frequently required in ``subgraph based'' applications. 
\footnote{In this paper we use the strong diameter guarantee for  routing, buy-at-bulk, and path reporting distance oracle.}

Sparse covers were introduced by Awerbuch and Peleg \cite{AP90} who showed that for every $k\in\N$, every $n$ vertex graph admits a strong $(4k-2,2k\cdot n^{\frac1k})$-\SPCS. %, and used it for routing and online tracking of mobile users.
Klein, Plotkin, and Rao \cite{KPR93} constructed a celebrated weak padded decomposition 
\footnote{Roughly, padded decomposition with padding parameter $\rho$ is a random partition into clusters of diameter at most $\Delta$ such that every ball of radius $\frac\Delta\rho$ is fully contained in a single cluster with probability at least $\frac12$. See \cite{Fil19padded}.} for $K_r$-minor free graphs with padding parameter $O(r^3)$ (later improved to $O(r^2)$ \cite{FT03}).
It is folklore that padded decomposition with padding parameter $\rho$ implies a $(\rho,O(\log n))$-\SPCS (by taking the union of $O(\log n)$ independent  partitions). Thus \cite{KPR93,FT03} implied a weak $(O(r^2),O(\log n))$-\SPCS for $K_r$-minor free graphs.
Krauthgamer \etal \cite{KLMN04} observed that one can use \cite{KPR93,FT03} padded decomposition to construct a weak $(O(r^2),2^r)$-\SPCS for $K_r$-minor free graphs (see \cite{Fil20scattering} for an explicit proof). This was the first sparse cover where all the parameters are independent from the cardinality of the vertex set $n$.
Busch \etal \cite{BLT14} used shortest path separators to obtain a strong $(8,f(r)\cdot\log n)$-\SPCS. \footnote{$f(r)$ is the enormous constant hiding in the Robertson Seymour \cite{RS03} structure theorem. Johnson \cite{Johnson87} estimated $f(r)\ge2\Uparrow(2\Uparrow(2\Uparrow\frac{r}{2}))+3)$ where $2\Uparrow t$ is the exponential tower function ($2\Uparrow0=1$ and $2\Uparrow t=2^{2\Uparrow(t-1)}$) \label{foot:RS}.}
Abraham \etal \cite{AGMW10} constructed a strong sparse cover. Specifically, they made some adaptations to \cite{KPR93} to obtain a strong $(O(r^2),2^{O(r)}\cdot r!)$-sparse cover scheme (not a \SPCS).
%\atodo{fix}
The final major piece in our story is a weak padded decomposition for $K_r$-minor free graphs with padding parameter $O(r)$ by Abraham \etal \cite{AGGNT19}. Which was later improved to strong padded decomposition with the same parameter by Filtser \cite{Fil19padded}. In particular, this implies a strong $(O(r),O(\log n))$-sparse cover scheme.
Interestingly, it was unknown how to turn this padded decomposition into a sparse cover with sparsity independent of $n$. Indeed, Filtser explicitly asked whether it is possible to construct $(O(r),g(r))$-sparse cover scheme \cite{Fil19padded}	.
The main result of this paper is an affirmative answer to this question. A decade after the publication of \cite{AGGNT14} we finally managed to turn this padded decomposition into a sparse cover.

\begin{restatable}[Cover for Minor Free Graphs]{theorem}{MinorCover}
	\label{thm:MinorFreeCover}
	Every $K_r$-minor free graph admits the following:
	\begin{itemize}
		\item Strong $\left(O(r),O(r^{2})\right)$-\SPCS.
%		\item Strong  $\left(O(1),O(1)^{r}\right)$-\SPCS.
		\item For $\eps\in(0,\frac12)$, strong  $\left(4+\eps,O(\frac{1}{\eps})^{r}\right)$-\SPCS.		
	\end{itemize}
\end{restatable}
The first bullet in \Cref{thm:MinorFreeCover} improves over the previous \SPCS state of the art \cite{FT03,KLMN04,Fil20scattering} in a threefold manner: (1) the padding is improved quadratically from $O(r^2)$ to $O(r)$, (2) the sparsity is improved exponentially from $2^r$ to $O(r^2)$, (3) the diameter guarantee is now strong. Alternately, the second bullet improves the padding parameter from $O(r^2)$ to $O(1)$ while keeping a similar sparsity.
 See \Cref{tab:Covers} for a comparison of ours and previous work.

\begin{table}[]
	\begin{tabular}{|l|l|l|l|l|l|}
		\hline
		\multicolumn{1}{|l|}{\textbf{Family}} & \textbf{Padding} & \textbf{Sparsity}                          & \textbf{$\SPCS$?} & \textbf{Diameter} & \textbf{Ref}       \\ \hline
		\multicolumn{1}{|l|}{General}         & $4k-2$           & $2k\cdot n^{\frac1k}$                      & yes               & strong            & \cite{AP90}               \\ \hline
		Planar&$32$&$18$&yes {\footnotesize (implicit)} &strong&\cite{BLT14}\\\hline
		\multirow{9}{*}{$K_r$-minor free}     & $O(r^3)$         & $O(\log n)$                                & yes               & weak              & \cite{KPR93}              \\ \cline{2-6} 
		& $O(r^2)$         & $O(\log n)$                                & yes               & weak              & \cite{FT03}               \\ \cline{2-6} 
		
		& $O(r)$         & $O(\log n)$                                & yes               & weak              & \cite{AGGNT19}               \\ \cline{2-6} 
		
		& $O(r)$         & $O(\log n)$                                & yes               & strong              & \cite{Fil19padded}               \\ \cline{2-6} 
		
		& $8$              & $f(r)\cdot\log n$    $^{(\ref{foot:RS})}$                      & yes {\footnotesize (implicit)}                 & strong           & \cite{BLT14}              \\ \cline{2-6} 
		& $O(r^2)$         & $2^r$                                      & yes               & weak              & \cite{FT03,KLMN04}             \\ \cline{2-6} 
		& $O(r^2)$         & $2^{O(r)}\cdot r!$                           & no                 & strong            & \cite{AGMW10}             \\ \cline{2-6} \cline{2-6} 
		%		& \cellcolor[HTML]{FFFFC7}$5r$             & \cellcolor[HTML]{FFFFC7}$O(r^2)$                                   & \cellcolor[HTML]{FFFFC7}yes               & \cellcolor[HTML]{FFFFC7}strong            & \cellcolor[HTML]{FFFFC7}\Cref{thm:MinorFreeCover} \\ \hline
		%		Treewidth $\tw$       & $6$              & $\poly(\tw)$                                 & ?                 & ?                 & ?                  \\ \hline
		
		& $O(r)$             & $O(r^2)$                                   & yes               & strong            & \multirow{2}{*}{\Cref{thm:MinorFreeCover}}  \\ \cline{2-5} 
		
		%	& $O(1)$             & $O(1)^r$                                   & yes               & strong            &  \\ \cline{2-5} 
		
		& $4+\eps$             & $O(\frac1\eps)^r$                                   & yes               & strong            & \\ \hline
		%Treewidth $\tw$       & $6$              & $\poly(\tw)$ &&&\\\hline 
		
		%		\multirow{2}{*}{Doubling}             & $2$              & $4^{\ddim}$                                & ?                 & strong?           & \cite{AGGM06}             \\ \cline{2-6} 
		%		& $O(t)$           & $O(2^{\frac{\ddim}{t}}\cdot d\cdot\log t)$ & yes               & strong            & \cite{Fil19padded}              \\ \hline
	\end{tabular}
	\caption{\small{Summery of new and previous work on sparse covers. }}
	\label{tab:Covers}
\end{table}

\subsection{Low dimensional metric embeddings into $\ell_\infty$}
% \usepackage{multirow}

Metric embedding is a map between two metric spaces that approximately preserves pairwise distances.
Given a (finite) metric space $(X,d_X)$, a map $\phi: V \to
\R^k$, and a norm $\|\cdot\|$, the \emph{contraction} and
\emph{expansion} of the map $\phi$ are the smallest $\xi, \rho \geq 1$,
respectively, such that for every pair $x,y\in X$,
\[ \frac1\xi\cdot d(x,y) \leq \| \phi(x) - \phi(y) \| \leq
\rho\cdot d(x,y)~~. \] The \emph{distortion} of the map is then $\xi \cdot \rho$. The expansion $\rho$ is also called the Lipschitz constant of the embedding $\phi$.
Metric embeddings into norm spaces were thoroughly studied \cite{Bou85,LLR95,Mat96,ABN11}, and have a plethora of applications.

There is a special interest for metric embeddings into $\ell_\infty$. 
From an algorithmic viewpoint, there is a significant advantage in the additional structure the norm space is providing. One example being nearest neighbor search (NNS) \cite{Indyk98,BG19}, where $\ell_\infty$ enjoys succinct data structures. 
NNS for many other spaces works by first embedding the space into $\ell_\infty$, and then using the $\ell_\infty$ NNS data structure to answer queries in the original metric space (see e.g. \cite{FI99,Indyk02}).
From a geometric view point, $\ell_\infty$ is a special norm as it is a universal host metric space. That is, every finite metric spaces embeds isometrically (that is with distortion $1$) into $\ell_\infty$ (the so called \frechet~embedding). 
However, there are $n$-point metric spaces of which every isometric embedding into $\ell_\infty$ requires $\Omega(n)$ dimensions \cite{LLR95}. 
It is desirable to have low dimensional metric embeddings, as these are much more useful for algorithmic design.
Matou{\v{s}}ek~\cite{Mat96} showed that for every integer $t\ge2$, 
every metric space embeds with distortion $2t-1$
into $\ell_\infty$ of dimension $O(n^{1/t}\cdot t\cdot \log n)$
(which is almost tight assuming the Erd\H{o}s girth conjecture \cite{Er64}). 
For distortion $O(\log n)$,
Abraham \etal \cite{ABN11} later improved the dimension to $O(\log n)$.

For more restrictive metric spaces better results are known. Linial \etal \cite{LLR95} showed that every $n$-point tree metric embeds isometrically into $\ell_\infty^{O(\log n)}$, while Neiman \cite{Neiman16} showed that every metric space with doubling dimension $d$ embeds into $\ell_\infty^{\eps^{-O(d)}\cdot\log n}$ with distortion $1+\eps$.
Krauthgamer \etal \cite{KLMN04} showed that every $n$-point $K_r$-minor free graph $G$ embeds into $\ell_\infty^{\tilde{O}(3^r\cdot \log n)}$ with distortion $O(r^2)$. 
Their embedding follows from a \SPCS based on \cite{KPR93,FT03}. However, their construction uses additional properties of that cover, and we cannot simply plug in our \SPCS to get an improved embedding.
We show that under very general conditions, \footnote{In fact the reduction hold in general with no condition, while the dimension slightly increases. See \Cref{lem:fromSparseCoverToEmbedding}.} one can use \SPCS in a black box manner to obtain a metric embedding into $\ell_\infty$ (see \Cref{thm:SPCStoEmbeddingNoAspect}). 
As a corollary, we improve both the distortion and dimension compared to \cite{KLMN04}. We than slightly tailor the embedding for minor free graphs to push the distortion down all the way to $3+\eps$. 
See \Cref{tab:Embeddings} for a comparison of new and old results.
\begin{restatable}[]{corollary}{CorEmbeddingMinorFree}
	\label{cor:EmbeddingMinor}
	Every $n$ vertex $K_r$-minor free graph embeds into $\ell_\infty^{O(r^{2}\cdot\log r\cdot\log n)}$ with distortion $O(r)$.
%	Every $n$ vertex $K_r$-minor free graph $G$ can be embedded into $\ell_\infty^{O(r^{2}\cdot\log r\cdot\log n)}$ with distortion $O(r)$. 
	%	Alternatively,  $G$ can be embedded into $\ell_\infty^{O(2^{O(r)}\cdot\log n)}$ with distortion $O(1)$.
\end{restatable}
%One can also use our $\left(4\cdot(1+\eps),O(\frac{1}{\eps})^{r}\right)$-\SPCS to obtain embedding into $\ell_\infty^{O(\frac{1}{\eps})^{r+1}\cdot\log\frac{1}{\eps}\cdot\log\frac{n}{\eps}}$ with distortion $8\cdot(1+\eps)$.
%However, for the case of minor-free graphs, we open the black box of \Cref{thm:SPCStoEmbeddingNoAspect} and obtain an improved embedding.
\begin{restatable}[]{theorem}{EmbeddingMinorFree}
	\label{thm:EmbeddingMinor3}
	For every $\eps\in(0,\frac12)$, every $n$ vertex $K_r$-minor free graph $G$ can be embedded into $\ell_\infty^{O(\frac{1}{\eps})^{r+1}\cdot\log\frac{1}{\eps}\cdot\log\frac{n}{\eps}}$ with distortion $3+\eps$.
\end{restatable}

\begin{table}[]
	\begin{tabular}{|l|l|l|l|}
		\hline
		\textbf{Family}                   & \textbf{Distortion} & \textbf{Dimension}            & \textbf{Ref} \\ \hline
		\multirow{3}{*}{General Metric}   & $1$                 & $n-1$                         & \frechet      \\ \cline{2-4} 
		& $2k-1$              & $O(k\cdot n^{\nicefrac1k}\cdot\log n)$       & \cite{Mat96}     \\ \cline{2-4} 
		&    $O(\log n)$                 & $O(\log n)$                           &         \cite{ABN11}     \\ \hline
		Tree&$1$&$\Theta(\log n)$&\cite{LLR95}\\\hline
		\multirow{3}{*}{$K_r$-Minor Free} & $O(r^2)$            &    $\tilde{O}(3^r)\cdot\log n$                           & \cite{KLMN04}         \\ \cline{2-4} 
		& $O(r)$              & $\tilde{O}(r^{2})\cdot\log n$ & \Cref{cor:EmbeddingMinor}   \\ \cline{2-4} 
%		& $O(1)$              & $O(1)^{r}\cdot\log n$ & \Cref{cor:EmbeddingMinor}   \\ \cline{2-4} 
		& $3+\eps$              & $\tilde{O}(\frac1\eps)^{r+1}\cdot\log n$
		 & \Cref{thm:EmbeddingMinor3}        \\ \hline
%		\multirow{2}{*}{Treewidth $\tw$} 
%		Treewidth $\tw$ & $O(\tw^2)$          &                               &              \\ \cline{2-4} 
%		Treewidth $\tw$& $13$                & $\poly(\tw)\cdot\log n$       & \Cref{cor:EmbeddingTreewidth}         \\ \hline
%		Pathwdith $\pw$                   &                     &                               &              \\ \hline
	\end{tabular}
	\caption{\small{Summery of new and previous work on metric embeddings into $\ell_\infty$.}}
	\label{tab:Embeddings}
\end{table}

\subsection{Oblivious Buy-at-Bulk}
Given a weighted graph $G=(V,E,w)$ and a \emph{canonical fusion function} $f:\N\rightarrow\R_{\ge0}$ (see \Cref{subsec:BuyAtBulk} for definition), in the oblivious buy-at-bulk problem, the goal is to pick a route $P_i$ for every possible demand pair $\delta_i=(s_i,t_i)\in{V\choose2}$.
Then, given a specific set of demands $A=\{\delta_1,\dots,\delta_k\}$, the cost of our oblivious solution $\cP=\{P_1,\dots,P_k\}$ is $\cost(\cP)=\sum_{e\in E}f(\varphi_e)\cdot w(e)$, where $\varphi_e$ is the number of paths in $\cP$ using $e$.
The solution is said to have approximation ratio $\rho$ if for every subset of demands, the induced cost of the oblivious solution is at most $\rho$ times the optimal solution.
Due to the concavity of the canonical fusion function $f$, it is advantageous for the chosen paths to intersect as much as possible.
The best known approximation for general graphs is $O(\log^2n)$ \cite{GHR06}, while for planar graphs $O(\log n)$ approximation is known \cite{SBI11}, which is tight \cite{IW91}.
Srinivasagopalan \etal \cite{SBI11} left it as an explicit open problem to 
``obtain efficient solutions to other related network topologies, such as minor-free graphs.'' More than a decade later, compared with general graphs, nothing better for $K_r$-minor free graphs is known.
Using our sparse covers on top of a black box reduction from \cite{SBI11}, we obtain the following tight result (\cite{IW91}):
\begin{restatable}[]{corollary}{BuyAtBulk}
	\label{cor:MinorBuyAtBulk}	
	For every $n$-vertex weighted $K_r$-minor free graph $G=(V,E,w)$ admits an efficiently commutable solution to the oblivious buy-at-bulk problem with approximation ratio $O(r^6\cdot\log n)$. Furthermore, the solution is also oblivious to the concave faction $f$.
\end{restatable}

\subsection{Further Applications}
\paragraph*{Sparse partition and universal TSP / Steiner tree}
Given a weighted graph $G=(V,E,w)$, a $(\alpha, \tau,\Delta)$-sparse partition is a partition $\cC$ of $V$ into clusters with weak diameter at most $\Delta$, such that every ball of radius $\frac{\Delta}{\alpha}$ intersects at most $\tau$ clusters from $\cC$.
$G$ admits a $(\alpha, \tau)$-sparse partition scheme if it admits $(\alpha, \tau,\Delta)$-sparse partition for every $\Delta>0$. Using our \Cref{thm:MinorFreeCover}, we construct $\left(O(r),O(r^2)\right)$-sparse partition scheme for $K_r$ minor free graphs, improving over the previous state of the art of  $\left(O(r^2),2^r\right)$-sparse partition scheme \cite{JLNRS05,Fil20scattering}. See \Cref{cor:sparsePartition} in \Cref{subsec:SparsePartition} for further details.

In the universal TSP problem, we are given a metric space $(X,d_X)$ and the goal is to choose a single permutation $\pi$ (a universal TSP) of $X$, such that given a subset $S\subseteq X$, we visit the points in $S$ w.r.t. the order in $\pi$. This is the induced TSP tour by $\pi$. The permutation $\pi$ has stretch $\rho$, if the length of the induced tour for every subset $S$ is at most $\rho$ times larger than the optimal tour for $S$. There is a general reduction from sparse partition scheme to universal TSP. Using this reduction and our sparse partitions (\Cref{cor:sparsePartition}), given a shortest path metric of an $n$ point $K_r$ minor free graph, we construct universal TSP with stretch $O(r^4)\cdot\log n$ (\Cref{cor:UST}), exponentially improving the dependence on $r$ compared with previous results ($O(1)^r\cdot\log n$). 
The same phenomena occurs also for the universal Steiner tree problem.
See \Cref{subsec:UST} for further details.

\paragraph*{Name Independent Routing}
Here we are given an unweighted graph $G=(V,E)$ where the names (and ports) of all nodes are fixed. The goal is to design a compact routing scheme that will allow sending packages in the network, where routing decisions are made using small local routing tables, and the resulting routing paths are approximate shortest paths. This regime is considered more challenging and practical from the regime where nodes names and ports could be chosen by the routing scheme designer. 
Given a hereditary graph family that has $\alpha$-orientation (see \Cref{subsec:routing}), where each graph possess strong $(\tau,\beta)$-sparse cover scheme, Abraham \etal \cite{AGMW10} constructed name independent compact routing scheme with stretch $O(\beta)$, $O(\log n+\log\tau)$-bit headers and  tables of $O(\frac{\log^{3}n}{\log\log n}+\alpha\cdot\log n)\cdot\tau\cdot\log D$ bits. We use our strong sparse covers (\Cref{thm:MinorFreeCover})  to construct  name independent compact routing scheme significantly improving over previous work. See \Cref{cor:Labeling}, and \Cref{tab:routing} for a comparison. 

%These results cannot be extended to the name-independent
%domain since in that case it is known that a stretch of $3$ is required for unweighted
%trees if less than $o(n)$ space is used.
%
%This is the citation for the LB
%C. Gavoille and M. Gengler. Space-efficiency for
%routing schemes of stretch factor three. Journal of
%Parallel and Distributed Computing, 61(5):679–687,
%2001.
%
%Abraham et al Compact Name-Independent Routing with Minimum Stretch
%Get stretch $3$ with routing tables of $\tilde{O}(\sqrt{n})$.
%
%$O(n\cdot\log n)$ bits are used [3].
%[3] - Abraham, I., Gavoille, C., Malkhi, D., Nisan, N., Thorup, M.: Compact name-independent routing
%with minimum stretch

\paragraph*{Path reporting distance oracle}
A path reporting distance oracle (\PRDO) for a weighted graph $G=(V,E,w)$ is a succinct data structure that given a query $\{x,y\}$, efficiently returns an approximate $x$-$y$ shortest path $P$. 
We say that a \PRDO has stretch $k$ and query time $t$, if for every query $(x,y)$, the oracle returns a path $P$ of weight at most $k\cdot d_G(x,y)$ in $O(|P|)+t$ time.
\PRDO were first studied for general graphs. Later, Elkin \etal \cite{ENW16} constructed a \PRDO for $K_r$-minor free graphs based on strong sparse covers. We plug in our strong sparse cover from \Cref{thm:MinorFreeCover} to obtain improvements in both space and stretch, see \Cref{cor:PathReporting} (and \Cref{subsec:PathReporting}).

\subsection{Related Work}\label{sec:related}
We provided background on sparse covers in the introduction. We refer to the cited papers for additional background on sparse partitions, \UTSP, \UST, routing and distance oracles. Here we provide additional background on metric embeddings in order to put our results (\Cref{cor:EmbeddingMinor}, \Cref{thm:EmbeddingMinor3}) in a wider context. 
We begin with metric embeddings into $\ell_p$ spaces.
Every $n$ point metric space embeds into $\ell^{O(\log n)}_2$ with distortion $O(\log n)$ \cite{Bou85}, which is also tight \cite{LLR95}. Planar graphs, and more generally fixed minor free graphs, embed into $\ell^{O(\log n)}_2$ with distortion $O(\sqrt{\log n})$ \cite{Rao99,AFGN22}, which is also tight \cite{NR02}. The big open question here is regarding the embedding of such graphs into $\ell_1$. The upper bounds are the same as for $\ell_2$, while the only lower bound is $2$ \cite{LR10}. A long standing conjecture by Gupta \etal \cite{GNRS04} states that every graph family excluding a fixed minor, and in particular planar graphs, can be embedded into $\ell_1$ with constant distortion.
Some partial progress for planar graph was made in the cases where we care only about vertices laying on a small number of faces \cite{KLR19,Fil20}, or only about vertex pairs laying on the same face \cite{OS81,Kumar22}.

Refined notion of distortion in metric embeddings were studied, such as scaling distortion \cite{ABN11,BFN19}, and terminal/prioritized distortion \cite{EFN17,EFN18}. In particular, there been study of prioritized low dimensional embeddings into $\ell_\infty$ \cite{FGK23,EN22}.
Online metric embeddings into normed spaces were also studied \cite{IMSZ10,NR20,BFT24}.

A different venue of research is metric embeddings into trees, or more generally low treewidth graphs.
Every $n$-point metric space stochastically embeds into trees with expected distortion $O(\log n)$  \cite{Bar96,FRT04}. This result is tight even when the metric is the shortest path metric of a planar graph \cite{AKPW95}.
Every planar graph with diameter $\Delta$ can be (deterministically) embedded into a graph with treewidth $\tilde{O}(\eps^{-3})$ with additive distortion $\eps\cdot\Delta$ \cite{FKS19,FL22tw,CCLMST23Planar}.
Every $K_r$-minor free graphs stochastically embeds into graphs with treewidth $f(r)\cdot O(\frac{\log\log n}{\eps})^2$ with expected additive distortion $\eps\cdot\Delta$ \cite{CFKL20,FL22tw}. Clan embeddings and Ramsey-type embeddings of $K_r$-minor free graphs were also studied \cite{FL22}.
Finally, recently it was shown that every $K_r$-minor free graphs stochastically embeds into graphs with treewidth $f(r)\cdot\tilde{O}(\frac1\eps)\cdot\poly(\log (n\cdot\Phi))$ with multiplicative expected distortion $1+\eps$ \cite{CLPP23} (here $\Phi$ is the aspect ratio).

\section{Preliminaries}\label{sec:perlims}
$\tilde{O}$ notation hides poly-logarithmic factors, that is $\tilde{O}(g)=O(g)\cdot\polylog(g)$. All logarithms are at base $2$ (unless specified otherwise), $\ln$ stand for the natural logarithm. Given a set $A$, ${A\choose 2}=\left\{\{x,y\}\mid x,y\in A,x\ne y\right\}$ denotes all the subsets of size $2$. For a number $k\in\N$, $[k]=\{1,2,\dots,k\}$.

We consider connected undirected graphs $G=(V,E,w)$ with edge weights
$w: E \to \R_{\ge 0}$. We say that vertices $v,u$ are neighbors if $\{v,u\}\in E$. Let $d_{G}$ denote the shortest path metric in $G$.
$B_G(v,r)=\{u\in V\mid d_G(v,u)\le r\}$ is the closed ball of radius $r$ around $v$. For a vertex $v\in V$ and a subset $A\subseteq V$, let $d_{G}(x,A):=\min_{a\in A}d_G(x,a)$,
where $d_{G}(x,\emptyset)= \infty$. For a subset of vertices
$A\subseteq V$, $G[A]$ denotes the induced graph on $A$,
and $G\setminus A := G[V\setminus A]$.
The \emph{diameter} of a graph $G$ is $\diam(G)=\max_{v,u\in V}d_G(v,u)$, i.e. the maximal distance between a pair of vertices.
Given a subset $A\subseteq V$, the \emph{weak}-diameter of $A$ is $\diam_G(A)=\max_{v,u\in A}d_G(v,u)$, i.e. the maximal distance between a pair of vertices in $A$, w.r.t. to $d_G$. The \emph{strong}-diameter of $A$ is $\diam(G[A])$, the diameter of the graph induced by $A$. 

The \emph{aspect ratio} of a metric space $(X,d_X)$ is usually defined as the ratio between the maximum and minimum distances. 
However, here we mainly work with the shortest path distance in graphs, where we allow $0$-weights. Formally, here the shortest path metric is actually a pseudometric. Accordingly, we will define aspect ratio of a graph $G=(V,E,w)$ as the ratio between the maximum distance to the minimum non zero distance: 
$\Phi(G)=\frac{\max_{u,v\in V}d_G(u,v)}{\min_{u,v\in V~\rm{s.t.~}d_{\footnotesize G}(u,v)>0}d_G(u,v)}$.

A graph $H$ is a \emph{minor} of a graph $G$ if we can obtain $H$ from
$G$ by edge deletions/contractions, and isolated vertex deletions.  A graph
family $\mathcal{G}$ is \emph{$H$-minor-free} if no graph
$G\in\mathcal{G}$ has $H$ as a minor.
Some examples of minor free graphs are planar graphs ($K_5$ and $K_{3,3}$ minor-free), outer-planar graphs ($K_4$ and $K_{3,2}$ minor-free), series-parallel graphs ($K_4$ minor-free) and trees ($K_3$ minor-free).

The $\ell_\infty$-norm of a vector $x=(x_1,\dots,x_k)\in \R^k$ is $\Vert
x \Vert_{\infty} : =\max_{i\in[k]}|x_{i}|$.
An embedding from a metric space $(X,d_X)$ into $\ell_\infty$ is a function $f:X\rightarrow\R^k$. The embedding $f$ has distortion $c\cdot t$ if for every $x,y\in X$, $\frac1c\cdot d_X(x,y)\le \|f(x)-f(y)\|_\infty\le t\cdot d_X(x,y)$. $t$ is the \emph{expansion} (also known as a \emph{Lipschitz} constant)  of $f$, while $c$ is that \emph{contraction} of $f$. An embedding with distortion $1$ (where $c=t=1$) is called \emph{isometric}.
Embedding $f:X\rightarrow \ell_\infty^k$ can be viewed as a collection of embeddings $\{f_i\}_{i=1}^k$ into the line $\R$, where $f_i(x)$ equals to the $i$'th coordinate of $f(x)$. We will also denote $(f(x))_i=f_i(x)$.
Using this notation, $f$ has expansion $t$ if for every $x,y\in X$ and $i\in[k]$, $|f_i(x)-f_i(y)|\le t\cdot d_X(x,y)$. Similarly, $f$ has contraction $c$ if for every $x,y\in X$ there is some  $i\in[k]$ such that $|f_i(x)-f_i(y)|\ge \frac1c\cdot d_X(x,y)$. In this case, we will say that the pair $x,y$ is \emph{satisfied} by the coordinate $i$.

\section{Technical Overview}
\paragraph*{Cop Decomposition.} Abraham \etal \cite{AGGNT19} constructed a padded decomposition for $K_r$-minor free graphs based on the cops-and-robbers game \cite{And86}. Fix the scale parameter $\Delta>0$. The process works as follows: pick arbitrary $x_1$, and let the ball $B_G(x_1,r_1)$ be the first cluster $\eta_1$, where $r_1\in[0,\Delta]$ is sampled using truncated exponential distribution. To construct the second cluster, pick an arbitrary connected component $C_2$ of $G\setminus \eta_1$, and arbitrary $x_2\in C_2$. Let $T_2$ be a shortest path from $x_2$ to some vertex $y$ with a neighbor in $\eta_1$. Then the second cluster $\eta_2=B_{G[C_2]}(T_2,r_2)$ is a ball around $T_2$ in the graph induced by the connected component, where the radius $r_2\in[0,\Delta]$
is sampled using truncated exponential distribution. 
In general, suppose that we already constructed clusters $\eta_1,\dots,\eta_{k-1}$. 
Let $C_k$ be an arbitrary connected component of $G\setminus\cup_{i<k}\eta_i$, and $x_k\in C_k$ arbitrary vertex. 
Let $\cK_{C_k}\subseteq \{\eta_1,\dots,\eta_{k-1}\}$ be all the previously created clusters $\eta_i$, such that there is an edge from $\eta_i$ to $C_k$. 
Let $T_k$ be a shortest path tree in $C_k$, with $x_k$ as a root, and such that for every $\eta_i\in\cK_{C_k}$, there is an edge from a vertex in $T_k$ to $\eta_i$. In particular, $T_k$ will have at most $|\cK_{C_k}|$ leaves. 
The $k$'th cluster $\eta_k=B_{G[C_k]}(T_k,r_k)$ is a ball around $T_k$ in the induced graph $G[C_k]$, where $r_k\in[0,\Delta]$ is sampled using truncated exponential distribution. See \Cref{fig:BufCD} for illustration.

One can run the cop decomposition on general graphs without getting any interesting structure. What makes it particularly interesting for $K_r$-minor free graphs is the fact that the size of the set $\cK_{C_k}$ of neighboring previously created clusters is always bounded by $r-2$. Indeed, one can argue that if $|\cK_{C_k}|\ge r-1$, than one can contract all the internal edges inside each cluster in $\cK_{C_k}$, and the connected component $C_k$, and obtain $K_r$ as a minor. Thus $T_k$ is a shortest path tree (w.r.t. $G[C_k]$) with at most $r-2$ leaves, and the cluster $C_k$ is a ball of radius at most $\Delta$ around this tree. We will call each such cluster a \emph{supernode}, and the shortest path tree $T_k$ it's \emph{skeleton}. In addition, we construct a tree $\cT$ over the supernodes. Here each supernode $\eta_k$ will be the child in $\cT$ of the last created supernode $\eta_i\in \cK_{C_k}$. Note that the only outgoing edges from $\eta_k$ are either to its descendants or ancestors in $\cT$. Furthermore, $\eta_k$ have at most $r-2$ ``neighbor'' ancestors. We will call the connected component $C_k$ where we constructed the supernode $\eta_k$ (with skeleton $T_k$) the \emph{domain} of $\eta_k$, denoted $\dom(\eta_k)$. Note that the vertices in $\dom(\eta_k)$ will either belong to $\eta_k$, or to the descendants of $\eta_k$ w.r.t. $\cT$. See \Cref{fig:BufCD} for illustration.

Due to the truncated exponential distribution, Abraham \etal \cite{AGGNT19} showed that a small ball $B_G(x,\gamma\Delta)$ is likely to be fully contained in a single supernode. \footnote{Specifically, using a sophisticated argument, based on a potential function, Abraham \etal \cite{AGGNT19} showed that with probability at least $e^{-O(r)\cdot\gamma}$, the ball $B_G(x,\gamma\Delta)$ is fully contained in a single supernode.}
Unfortunately, the supernodes $\{\eta_i\}_{i\ge1}$ do not have a bounded diameter.
Nevertheless, following \cite{Fil19padded}, using the skeleton $T_k$, one can partition each supernode $\eta_k$ into clusters of diameter $O(\Delta)$, while cutting each small ball only with a small probability. Combining these two processes together, one obtains a strong $(O(r),\Omega(\frac{1}{r}),O(\Delta))$-padded decomposition. 
However, it is unclear if it is possible to use the cop decomposition to create a sparse cover. Indeed, the ``dependence tree'' $\cT$ does not have a bounded depth, and the entire process looks very chaotic. Indeed, every small change in the sampling of the radii leads to a completely different outcome. 
In contrast, the \cite{KPR93} (as well as \cite{AGMW10}) clustering process had only a depth of $r$, and thus \cite{KLMN04} enumerated all the possible choices in the process leading to a sparse cover of exponential sparsity.

\paragraph*{Buffered Cop Decomposition.} In a recent work, Chang \etal \cite{CCLMST24} obtained a new ``separation''-property in the cop-decomposition. Instead of growing a ball with a random radius around the skeleton $T_k$, Chang \etal constructed the supernode $T_K$ deterministically.
The new separation property is the following, consider a supernode $\eta$, and a vertex $v\in\dom(\eta)$ such that $v$ belongs to a supernode $\eta'$, which is descendant of $\eta$, but there is no edge from $\eta$ to $\eta'$. Then $d_{G[\dom(\eta)]}(\eta,v)>\frac{\Delta}{r}=\gamma$. In other words, for every descendant $\eta'$ of $\eta$, either they are neighbors, or every path in $\dom(\eta)$ from $\eta$ to $\eta'$ is of length at least $\gamma$.   \footnote{Roughly speaking, \cite{CCLMST24} begin with the supernode $\eta_k$ being equal to the skeleton $T_k$. Then, as the algorithm progresses, each time the buffer property is violated we add the violating vertices to a previously created supernode. One can argue that the depth of this process is bounded by $r$ (once for each neighboring ancestor supernode), and thus the cluster vertices are all within $\Delta$ distance from the skeleton.}
Chang \etal called the new partition \emph{Buffered Cop Decomposition}, because there is now a buffer between non-neighboring clusters. They used the new buffered cop decomposition to construct a \emph{shortcut-partition} (which is a generalization of scattering partition \cite{Fil20scattering}). Roughly, a shortcut-partition is a partition of the vertices into clusters of diameter at most $\eps\cdot D$, such that for every pair of vertices $u,v$ at distance at most $D$, there is an approximate shortest path going through at most $O_r(\frac1\eps)$ clusters.
Chang \etal used their shortcut-partitions to construct tree covers \cite{BFN19Ramsey,CCLMST23Planar}, distance oracles \cite{Tho04,Kle02,AG06}, to solve the Steiner point removal problem \cite{Fil19SPR,Fil20scattering,KKN15,Cheung18,FKT19}, and to construct additive embedding of apex minor free graphs into low treewidth graphs \cite{FKS19,CFKL20,FL21,FL22tw,CCLMST23Planar}.

\paragraph*{Sparse Covers.} The starting point of this paper is the buffered cop decomposition of \cite{CCLMST24}. We begin by observing some additional properties. Let $\overrightarrow{G}_{\cC}$ be a digraph where the supernodes are the vertices, and there is a directed  edge from a supernode $\eta$ to its ancestor $\eta'$ (w.r.t. $\cT$) iff there is an edge between vertices in $\eta$ and $\eta'$.
$\overrightarrow{G}_{\cC}$ is a DAG (directed acyclic graph) with maximum out-degree $r-2$, and it has additional crucial property: if $v$ has outgoing edges towards $u,z$ than there has to be an edge between $u$ and $z$ (in one way or another). We call a graph with this property a \emph{transitive DAG} (see \Cref{def:transitiveDAG}). In transitive DAG's, neighboring vertices tend to share many of their neighbors.
Denote by $B_{\overrightarrow{G}_{\cC}}(\eta,q)$ the set of supernodes towards which there is a directed path from $\eta$ of length at most $q$.
We use the transitive DAG property to show that the size of  $B_{\overrightarrow{G}_{\cC}}(\eta,q)$ is bounded by ${r+q\choose r}$ (\Cref{lem:diballs}). Note that crucially, for $q\ge r$, 
${r+q\choose r}\approx O(q)^r$
the growth rate is sub-exponential in $q$.
Next, we generalize the buffer property to argue that for every ancestor supernode $\eta'$ of $\eta$ such that $\eta'\notin B_{\overrightarrow{G}_{\cC}}(\eta,2q+1)$ it holds that the distance from every vertex $v\in\eta$ to $\eta'$ in $\dom(\eta')$ is at least $(q+1)\cdot\frac{\Delta}{r}$ (\Cref{lem:BufferExtended}).
Combining these two properties together, it follows that for every vertex $v$ there are at most ${r+q\choose r}$ ancestor supernodes at distance $q\cdot\frac\Delta r$.

Similar to the process of padded decomposition \cite{AGGNT19,Fil19padded}, our sparse cover is constructed in two steps. First we cover the vertices using \emph{enlarged} supernodes, and then we separately cover each enlarged supernode.
Fix $q\ge0$, and for every supernode $\eta$, let $\hat{\eta}=B_{G[\dom(\eta)]}(\eta,q\cdot\frac\Delta r)$ be all the vertices at distance at most $q\cdot\frac\Delta r$ from $\eta$ in $\dom(\eta)$. In particular, a vertex $v\in\eta$ can join only to the enlarged supernodes which are ancestors of $\eta$ at distance at most $q\cdot\frac\Delta r$.
Consider the ball $B=B_G(v,\frac{q}{2}\cdot\frac{\Delta}{r})$, and let $\eta$ be the first supernode containing some vertex from $B$. By the minimality of $\eta$, and the triangle inequality, it will follow that the ball $B\subseteq\hat{\eta}$ is contained in the enlarged supernode. From the other hand, due to the properties discussed above, each vertex $v$ will belong to at most ${r+q\choose r}$ enlarged supernodes. Thus we get both the sparsity and the padding properties we wanted. The only missing property at this point is the bounded diameter. Next we cover each enlarged supernode $\hat{\eta}$ using the skeleton $T_\eta$. $T_\eta$ consist of at most $r$ shortest path. 
We go over these shortest paths, and choose a $\Delta$-net $N$. Specifically, a set such that every two net points are at distance at least $\Delta$, and every point has a net point at distance at most $\Delta$. Fix $R=2\Delta+q\cdot\frac\Delta r$. Due to the properties of shortest paths, every point has at most $O(r+q)$ net points at distance $2R$.
Now taking all the balls of radius $R$ around net points provides us the desired sparse cover for the enlarged supernode. Taking the union of all the covers for all the enlarged supernodes we obtain the sparse cover for the graph.
Fixing $q=1$ we obtain padding $O(r)$ and sparsity $O(r^2)$, while by taking $q=\Theta(\frac r\eps)$ we obtain padding $4+\eps$ and sparsity $O(\frac{1}{\eps})^r$.

\paragraph*{Metric embedding into $\ell_\infty$.}
Our metric embedding into $\ell_\infty$ is based on our \SPCS. The first to construct a sparse cover based metric embedding was Rao \cite{Rao99}, who used the \cite{KPR93} padded decomposition to embed $K_r$ minor free graphs into $\ell_2$ with distortion $O(r^3\cdot\sqrt{\log n})$. Given a partition $\cP$, for a vertex $v$ belonging to cluster $v\in C_v\in\cP$, let $\partial_{\cP}(v)=d_X(v,V\setminus C_v)$ be the distance between $v$ to the boundary of the cluster $C_v$. Note that if $v$ is padded $B_G(v,\frac{\Delta}{\beta})\subseteq C_v$, then $\partial_{\cP}(v)> \frac{\Delta}{\beta}$.
For each cluster $C\in\cP$, sample $\alpha_C\in\{\pm1\}$ u.a.r. . For every vertex $v$, send $v$ to $\alpha_{C_v}\cdot \partial_{\cP}(v)$. By the triangle inequality, it follows that the expansion is at most $2$. But for which vertex pairs can we guarantee small contraction?

Consider a pair $u,v$ at distance $d_G(v,u)\in(\Delta,2\Delta]$, and suppose that $\cP$ has diameter at most $\Delta$, and $v$ is padded in $\cP$. Then $u$ and $v$ must belong to different clusters $C_u,C_v$. If it so happened that $\alpha_{C_u}\ne \alpha_{C_v}$, then $\left|\alpha_{C_{v}}\cdot\partial_{\cP}(v)-\alpha_{C_{u}}\cdot\partial_{\cP}(u)\right|=\partial_{\cP}(v)+\partial_{\cP}(u)\ge\frac{\Delta}{\beta}\ge\frac{d_{G}(u,v)}{2\beta}$, and we obtain a bound on the contraction!
Rao took $O(\log n)$ independent samples of the coefficients $\{\alpha_C\}_{C\in\cP}$, and gets the contraction guarantee in a constant fraction of the samples. Next, Rao also took $O(\log n)$ independent partitions to get that $v$ is padded in a constant fraction of them. Finally, Rao sampled partitions for all possible distance scales, concatenated the resulting embeddings of them all, and obtained the desired  $O(r^3\cdot\sqrt{\log n})$ distortion.
Assuming all the distances are in $[1,\poly(n)]$, there are $O(\log n)$ different distances scales, and the resulting dimension is $O(\log^3n)$: one $\log$ for the number of scales, one $\log$ to sample many partitions for each scale, and one $\log$ to sample the coefficients $\{\alpha_C\}_{C\in\cP}$. Nevertheless, in $\ell_2$, using dimension reduction \cite{JL84}, one can reduce the dimension to $O(\log n)$ without significantly increasing the distortion.

The focus of our paper is embeddings into $\ell_\infty$. One can repeat Rao's embedding exactly as is, and get embedding into $\ell_\infty^{O(\log^3n)}$ with distortion $O(r^3)$ (or $O(r)$ using the improved padding parameter from \cite{AGGNT19}). 
Unfortunately, there is no general dimension reduction in $\ell_\infty$ (ala \cite{JL84}). Nevertheless, the dimension still can be dramatically improved. First, observe that when embedding into $\ell_\infty$, we don't need to succeed on a constant fraction of the partitions (or the $\alpha$ coefficients), it is enough to be successful only once! Thus, instead taking $O(\log n)$ independent samples from a padded decomposition, one can use a \SPCS. 
Indeed, Krauthgamer \etal \cite{KLMN04} constructed an $(O(r^2),3^r)$-\SPCS for $K_r$ minor free graphs. Using this \SPCS immediately leads to an embedding into $\ell_\infty^{\tilde{O}(3^r)\cdot\log^2 n}$ with distortion $O(r^2)$. 
To remove additional $\log n$ factor, \cite{KLMN04} used an additional property of the \cite{KPR93} based \SPCS that does not holds in general: it is possible to create $3^r$ partitions such that for every pair $u,v$, there will be a single partition where both $u$ and $v$ will be padded \textbf{simultaneously}. 
\cite{KLMN04} heavily relied on this property, while our \SPCS (and actually all the others as well) lacking it. Hence we cannot apply \cite{KLMN04} as is.

Our solution follow similar lines to \cite{KLMN04}, but avoids using the additional special structure of \cite{KPR93}. Consider a graph $G$ with a $(\beta,\tau)$-\SPCS. First, for every scale $\Delta_i=\rho^i$, for $\rho=O(\frac{\beta}{\eps})$, create $\tau$ partitions $\cP_i^1,\cP_i^2,\dots,\cP_i^\tau$, all with diameter $\Delta_i$, and such that every vertex is $\beta$-padded in one of them (that is $\forall v,~B_G(v,\frac{\Delta}{\beta})$ is contained in some cluster). Next, 
create, laminar partitions that closely resemble the original partitions. Specifically, we create  $\{\tilde{\cP}_i^1\}_{i\in\Z},\{\tilde{\cP}_i^2\}_{i\in\Z},\dots,\{\tilde{\cP}_i^\tau\}_{i\in\Z}$, where for every $i,j$, $\cP_i^j$ refines $\cP_{i+1}^{j}$, $\cP_i^j$ has radius at most $(1+\eps)\cdot\Delta_i$, and for every $v\in V$, and $i\in \Z$, $B_G(v,\frac{\Delta}{(1+\eps)\cdot\beta})$ is fully contained in a cluster of one of $\tilde{\cP}_i^1,\dots,\tilde{\cP}_i^\tau$. In other words, for every $v\in V$, and $i\in\Z$, $\max_{j\in[\tau]}\partial_{\cP_i^j}(v)>\frac{\Delta_i}{(1+\eps)\cdot\beta}$.
Similar laminar partitions were also created in \cite{KLMN04}.

Next, our goal is to embed w.r.t. each laminar partition $\{\tilde{\cP}_i^j\}_{i\in\Z}$ independently, such that for every pair $u,v$ which is separated in partition $\tilde{\cP}_i^j$ it will hold that $\|f(u)-f(v)\|_\infty\ge\partial_{\tilde{\cP}_i^j}(v)+\partial_{\tilde{\cP}_i^j}(u)$.~\footnote{For comparison, \cite{KLMN04} used the simultaneous padding property of \cite{KPR93} and only guaranteed\\ $\|f(u)-f(v)\|_\infty\ge\min\left\{\partial_{\tilde{\cP}_i^j}(v),\partial_{\tilde{\cP}_i^j}(u)\right\}$.
%	 In fact, the proof in \cite{KLMN04} in insufficient even for that. Nevertheless, the authors of \cite{KLMN04} communicated to us a corrected proof of their claim.
 }
Note that given such embedding for all the laminar partitions, we can concatenate them all to obtain the desired distortion. Indeed, constant expansion follows by triangle inequality, while for every $u,v$,
$\|f(u)-f(v)\|_\infty\ge\max_{i,j}\left(\partial_{\tilde{\cP}_i^j}(v)+\partial_{\tilde{\cP}_i^j}(u)\right)=\Omega(\frac{d_G(u,v)}{\beta})$.
To create the embedding w.r.t. to the laminar partition $\{\tilde{\cP}_i^j\}_{i\in\Z}$, let $i_{0}\in\Z$ be a the maximum index such that $\tilde{\cP}_{i_0}^j$ is partitioned into singletons, and let $k\in\N$ be the maximum such that $\tilde{\cP}_{i_0+k}^j$ is not the trivial partition into a single cluster $\{V\}$. Clearly it is enough to embed only w.r.t. the laminar partition  $\{\tilde{\cP}_i^j\}_{i=i_0}^{i_0+k}$. Let $\tilde{\cP}_{i_0+k}^j=\{A_1,\dots,A_m\}$ be the clusters in the top partition in our hierarchy. We create a prefix free code $\alpha:\tilde{\cP}_{i_0+k}^j\rightarrow\{\pm1\}^*$. Specifically, each cluster $A_q$ is assigned a string of $\pm1$ of length at most  $2\cdot\left\lceil \log\frac{|V|}{|A_q|}\right\rceil$. The strings of different clusters might be of different length. However, for every $A_q,A_{q'}$ there is an index $s$ such that $\alpha_s(A_q),\alpha_s(A_{q'})$ exist and differ. Then the embedding of each vertex $v\in A_q$ defined by concatenating the coordinates $\left(\alpha_1(A_q)\cdot\partial_{\tilde{\cP}_{i_0+k}^j}(v),\alpha_2(A_q)\cdot\partial_{\tilde{\cP}_{i_0+k}^j}(v),\dots\right)$ with an embedding created inductively for $A_q$ w.r.t. the induced partition by $\{\tilde{\cP}_i^j\}_{i=i_0}^{i_0+k-1}$. 
To bound the contraction, consider a pair $u,v$ and suppose that they are first separated in level $k'\in[i_0,i_0+k]$. Then the embeddings of $u$ and $v$ are ``aligned'' in scales $[k'+1,i_0+k]$, while in scale $k'$ they belong to respective clusters $A_v,A_u\in \tilde{\cP}_{k'}^j$. The respective codes $\alpha(A_v),\alpha(A_u)$ will differ in some coordinate $s$ and thus $\|f(u)-f(v)\|_{\infty}\ge\left|\alpha_{s}(A_{u})\cdot\partial_{\tilde{\cP}_{k'}^{j}}(u)-\alpha_{s}(A_{v})\cdot\partial_{\tilde{\cP}_{k'}^{j}}(v)\right|=\partial_{\tilde{\cP}_{k'}^{j}}(u)+\partial_{\tilde{\cP}_{k'}^{j}}(v)
$. To conclude a bound on the contraction it remains to observe that for every partition $\tilde{\cP}_{l}^j$ that refines $\tilde{\cP}_{k'}^j$ it holds that $\partial_{\tilde{\cP}_{k'}^j}(v)+\partial_{\tilde{\cP}_{k'}^j}(u)\ge \partial_{\tilde{\cP}_{l}^j}(v)+\partial_{\tilde{\cP}_{l}^j}(u)$.

The overall number of coordinates used for the embedding of the laminar partition $\{\tilde{\cP}_i^j\}_{i=i_0}^{i_0+k}$ is $2\cdot\lceil\log n\rceil+2(k+1)$, where $k$ have to be bounded by a logarithm of the aspect ratio $O(\log\Phi)$. If we started from a $(\beta,\tau)$-\SPCS, there are $\tau$ laminar partitions and thus the overall dimension is $O(\tau\cdot \log(n\cdot\Phi))$.
We remove the dependence on the aspect ratio using fairly standard techniques. Specifically, by observing that for far enough scales, we can use the same coordinate. For this to hold, when treating scale $\rho^i$, we need to ``contract'' all vertex pairs at distance at most $\frac{\rho^i}{n^2}$ (as otherwise a pair can accumulate error in unbounded number of scales). Hence we can remove the dependence on the aspect ratio only if the contracted graph still admits a \SPCS.

\section{Sparse Covers for Minor Free Graphs}

This section is devoted to proving a meta-theorem (\Cref{thm:ReductionBufferedToCovers}) that given a buffered cop decomposition, constructs a sparse cover in a black box manner. Afterwards, our main \Cref{thm:MinorFreeCover} will follow as a corollary of this meta theorem, and the buffered cop decomposition of $K_r$ minor free graphs from \cite{CCLMST24}.
We begin in \Cref{subsec:Buffered} with recalling the definition of buffered cop decomposition from \cite{CCLMST24}. Then, in \Cref{subsec:BufferedAdditionalProperties} we prove some properties of buffered cop decomposition (mainly bounding the possible number of supernodes at a certain distance, see \Cref{lem:diballs,lem:BufferExtended}).
Finally, in \Cref{subsec:SparseCoverCOnstruction}, we construct a sparse cover from a buffered cop decomposition.

\subsection{Buffered cop decomposition.}\label{subsec:Buffered}
%We begin be recalling the definition of buffered cop decomposition from \cite{CCLMST24}.
Let $G=(V,E,w)$ be a weighted graph.
A \emph{supernode $\eta\subseteq V$} 
with \emph{skeleton $T_\eta$} and \emph{radius~$\Delta$} is an induced subgraph $G[\eta]$ of $G$ containing a tree $T_\eta$ where every vertex in $\eta$ is within distance $\Delta$ of $T_\eta$ w.r.t. induced shortest path metric, that is $\eta=B_{G[\eta]}(T_\eta,\Delta)$.
%for some real number $\Delta$, where distance is measured with respect to $\eta$.
%
%We occasionally abuse notation and use $\eta$ to refer to the set of vertices in $\eta$, rather than the subgraph.
% 
A \emph{buffered cop decomposition} for $G$ is a partition of $G$ into vertex-disjoint supernodes, together with a tree~\emph{$\cT$} called the \emph{partition tree}, whose nodes are the supernodes of $G$. 
For any supernode~$\eta$, the \emph{domain $\dom(\eta)$} denotes the subgraph induced by the union of all vertices in supernodes in the subtree of $\cT$ rooted at $\eta$.
See \Cref{fig:BufCD} for an illustration of supernodes, their domain, and a buffered cop decomposition.

\begin{definition}\label{def:buffer-cop}
	A \emph{$(\Delta, \gamma, w)$-buffered cop decomposition}
	for $G$ is a buffered cop decomposition $\cT$ that satisfies the following properties:
	\begin{itemize}
		
		\item \textnormal{[Supernode radius.]} 
		Every supernode $\eta$ has radius at most $\Delta$.

		\item \textnormal{[Shortest-path skeleton.]} 
		For every supernode $\eta$,	let $\cA_\eta$ be the set of ancestor supernodes $\eta'$ of $\eta$ such that there is an edge from $\dom(\eta)$ to $\eta'$. Then $|\cA_\eta|\le w$. 
		The skeleton $T_\eta$ is an SSSP tree in $\dom(\eta)$, with at most $w$ leaves (the root is not counted). In particular, for every $\eta'\in\cA_\eta$, there is an edge from $T_\eta$ to $\eta'$. \footnote{In the original definition of buffered cop decomposition in \cite{CCLMST24} there were no explicit requirement for edges between $T_\eta$ to $\cA_\eta$. However it is holds in their construction.}
%		
%		\item \textnormal{[Shortest-path skeleton.]} 
%		For every supernode $\eta$,		
%		 the skeleton $T_\eta$ is an SSSP tree in $\dom(\eta)$, with at most $w$ leaves (the root is not counted). In particular, there are at most $w$ ancestor supernodes $\eta'$ such that there is a crossing edge from  $\eta'$ to $\eta$. \footnote{In the original definition of buffered cop decomposition in \cite{CCLMST24} the bound on neighboring ancestor clusters was not explicit. However, it did hold  where combined with their tree decomposition.}
		
		\item \textnormal{[Supernode buffer.]} 
		Let $\eta$ be a supernode, and let $\eta'$ be another supernode that is an ancestor of $\eta$ in the partition tree $\cT$.
		Then either $\eta$ and $\eta'$ are adjacent in $G$, or
		for every vertex $v$ in $\dom(\eta)$, we have $d_{\dom(\eta')}(v, \eta') > \gamma$.
		
	\end{itemize}
\end{definition}
\begin{figure}[t]
	\centering
	\includegraphics[width=.96\textwidth]{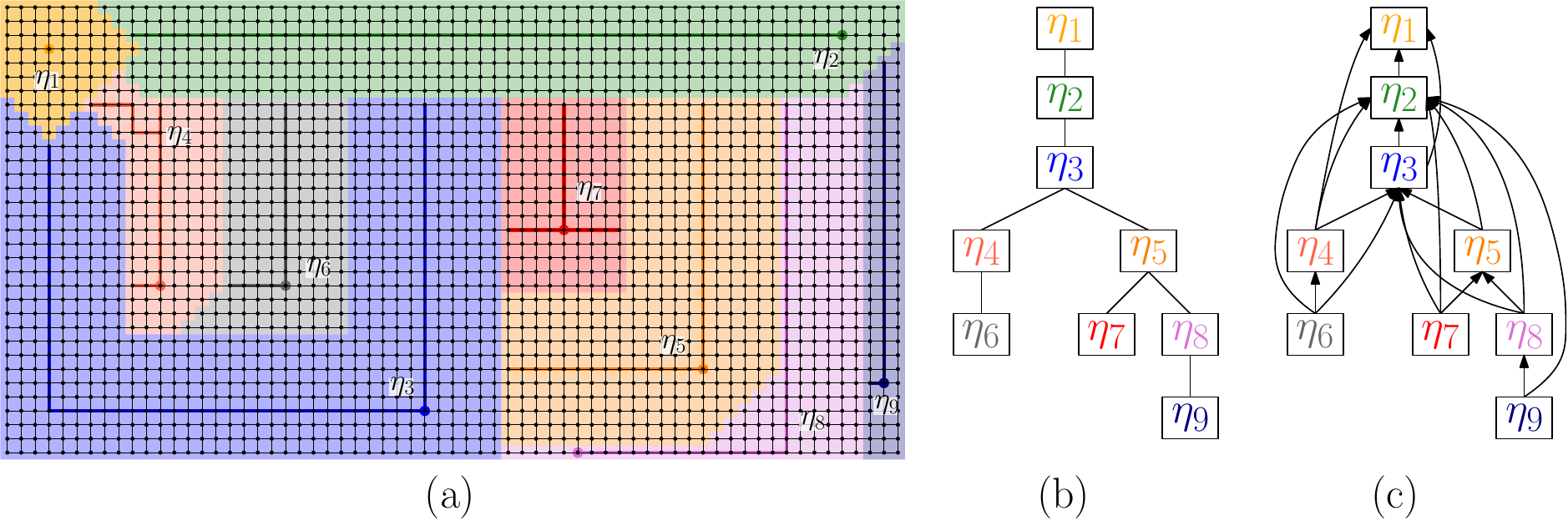}
	\caption{\footnotesize{(a) Illustration of an $(\Delta, \gamma, w)$-buffered cop decomposition of the unweighted grid graph together with (b) - the associated tree $\cT$. There are $9$ different supernodes $\eta_1,\dots,\eta_9$, all colored with different colors. Each supernode $\eta$ contains a shortest path tree $T_\eta$ (the bold lines) with at most $3$ leaves, where all the vertices $x\in\eta$ in the super node are at distance at most $\Delta=6$ from $T_\eta$. The domain of each supernode consist of all the supernodes in its subtree. For example $\dom(\eta_5)=G\left[\eta_5\cup\eta_7\cup\eta_8\cup\eta_9\right]$, and $\dom(\eta_3)=G\left[V\setminus(\eta_1\cup\eta_2)\right]$. As $\eta_3$ and $\eta_9$ are not adjacent, the distance from $\eta_3$ to any vertex in $\eta_9$ (w.r.t. $\dom(\eta_3)$) is at least $\gamma$. The associated digraph   $\overrightarrow{G}_{\cC}$ is illustrated in (c).
	}}
	\label{fig:BufCD}
\end{figure}

%We say that such a buffered cop decomposition $\cT$ has \emph{radius $\Delta$}, \emph{buffer $\gamma$}, and \emph{width $w$}. See \Cref{fig:buffered-cop} for an illustration, and \Cref{T:glossary} for a glossary of terminologies for the buffered cop decompositions.

Chang \etal \cite{CCLMST24} constructed a buffered cop decomposition for $K_r$-minor free graphs based on the classic cop decomposition (see \cite{And86,AGGNT19,Fil19padded}).

\begin{restatable}[Buffered cop decompsition]{theorem}{BuffCD}
	\label{thm:buffered}	
	Let $G$ be a $K_r$-minor-free graph, and let $\Delta$ be a positive number. Then $G$ admits a $(\Delta, \Delta/r, r-1)$-buffered cop decomposition (efficiently computable).
\end{restatable}

\subsection{Additional properties of buffered cop decomposition}\label{subsec:BufferedAdditionalProperties}
%\begin{wrapfigure}{r}{0.18\textwidth}
%	\begin{center}
%		\vspace{-10pt}
%		\includegraphics[width=0.9\textwidth]{fig/BufGdirected}
%		\vspace{-5pt}
%	\end{center}
%	\vspace{-15pt}
%\end{wrapfigure}
Consider a $(\Delta, \gamma, w)$-buffered cop decomposition $\cT$ of a graph $G$. In this subsection we will prove some additional properties of the buffered cop decomposition, that will later be useful in our construction of the sparse covers. Let $\overrightarrow{G}_{\cC}$ be a directed graph with the supernodes $\cC$ as vertices. For a supernode $\eta$, we add a directed edge from $\eta$ to any adjacent ancestor supernode $\eta'$. See \Cref{fig:BufCD} (c) for an illustration.
%illustration on the right for the digraph $\overrightarrow{G}_{\cC}$ w.r.t. the decomposition from \Cref{fig:BufCD}.
Note that $\overrightarrow{G}_{\cC}$ is a DAG as all edges are directed towards the ancestors w.r.t. $\cT$. $\overrightarrow{G}_{\cC}$ has another crucial property: suppose that there are outgoing edges from $\eta$ to $\eta_1$ and $\eta_2$, where $\eta_1$ is an ancestor of $\eta_2$ in $\cT$. Then there is an outgoing edge from $\eta_2$ to $\eta_1$. Indeed, as $\eta$ and $\eta_1$ are neighbors, there is an edge from $\dom(\eta)$ to $\eta_1$. As  $\dom(\eta)\subseteq\dom(\eta_2)$, there is an edge from $\dom(\eta_2)$ to $\eta_1$. 
In particular, $\eta_1\in\cA_{\eta_2}$ and hence there is an edge from $T_{\eta_2}$ (and in particular from $\eta_2$) to $\eta_1$. 
We call such a graph a \emph{Transitive DAG}.
\begin{definition}[Transitive DAG]\label{def:transitiveDAG}
	A digraph $\overrightarrow{G}=(V,\overrightarrow{E})$ is called transitive DAG if it contains no cycles, and for every $x,y,z\in V$ such that $(x,y),(x,z)\in \overrightarrow{E}$ it holds that either $(y,z)\in\overrightarrow{E}$ or $(z,y)\in \overrightarrow{E}$.
\end{definition}
Another important property of $\overrightarrow{G}_{\cC}$ is that the maximum out degree of a vertex is $w$. Indeed, this just follows from the definition of buffered cop decomposition. Give a digraph $\overrightarrow{G}$, and a vertex $v\in V$, the directed ball $B_{\overrightarrow{G}}(v,q)$ is the set of vertices $u$ such that there is a directed path from $v$ to $u$ in $\overrightarrow{G}$ of length at most $q$. $v$ is called the center of the directed ball, while $q$ is it's radius. We next bound the size of directed balls in a transitive DAG of bounded out degree.

\begin{lemma}\label{lem:diballs}
	Let $\overrightarrow{G}=(V,\overrightarrow{E})$ be a transitive DAG of maximum out degree $w$. Then for every $v\in V$, and $q\in\N$, $\left|B_{\overrightarrow{G}}(v,q)\right|\le {w+q\choose w}$.
\end{lemma}
\begin{proof}	
	The proof is by induction on $q$. The base case is when $q=0$, and indeed every directed ball of radius $0$ is of size ${w+0\choose w}=1$ (containing only the center). 
	Let $v_0,v_1,v_2,\dots,v_n$ be a topological ordering of $\overrightarrow{G}$ vertices. That is for every edge $(v_i,v_j)\in\overrightarrow{E}$ it holds that $i<j$.
%	We can assume w.l.o.g. that $v=v_0$, as no vertex left of $v$ in the topological order will join $B_{\overrightarrow{G}}(v,q)$.
	Denote $v=v_{i_0}$
	We can assume that $v_{i_0}$ has exactly $w$ outgoing edges, as otherwise we can just add artificial neighbors to the right of $v_n$, which will only increase the size of the directed ball. Denote by $v_{i_1},v_{i_2},\dots,v_{i_w}$ the endpoints of the $w$ edges going out of $v$. 
	For $s\in[0,w-1]$, let $I_s=\{v_{i_s},v_{i_s+1},\dots,v_{i_{s+1}-1}\}$ be the ``interval'' of vertices starting at $v_{i_s}$ and ending just before $v_{i_{s+1}}$. Denote also $I_w=\{v_{i_w},\dots,v_{n}\}$.
	
	For $s\in[0,w]$, denote by $B_s=I_s\cap B_{\overrightarrow{G}}(v_{i_0},q)$ the intersection of $I_s$ with $B_{\overrightarrow{G}}(v_{i_0},q)$. 
	Clearly $B_s=B_{\overrightarrow{G}[I_s]}(v_{i_s},q-1)$. That is, $v_{i_s}$ is the only neighbor of $v$ in $I_s$,  and as $\overrightarrow{G}$ is a DAG, all the paths from $v_{i_0}$ to vertices in $I_s$ must go though $v_{i_s}$.
	Next, we argue that for every vertex $u_j\in B_s$, $u_j$ has outgoing edges towards $v_{i_{s+1}},v_{i_{s+2}},v_{i_{w}}$. The proof is by induction w.r.t. the topological order. The base case is $v_{i_{s}}$, and it holds by the transitive DAG rule as $v_{i_0}$ has outgoing edges towards $v_{i_s},\dots,v_{i_w}$.
	By induction, consider $u_j\in B_s$ and suppose that the last edge on it shortest paths from $v_{i_s}$ to $u_j$ is $(u_{j'},u_j)$. By the inductive hypothesis, $u_{j'}$ has outgoing edges towards $v_{i_{s+1}},\dots,v_{i_w}$. According to the rule of the transitive DAG, $u_j$ also has outgoing edges towards $v_{i_{s+1}},\dots,v_{i_w}$, as required. 
	We conclude that the induced digraph $\overrightarrow{G}[B_s]$ has maximum outgoing degree of $s$ (as the maximum out degree in $\overrightarrow{G}$ is $w$ and each such vertex has already $w-s$ outgoing edges out of $B_s$).
	By the inductive hypothesis, we conclude $|B_s|\le{s+(q-1)\choose s}$.
	
	For convenience denote $B_0=\{v_{i_0}\}$ (one can think of it as the ball in the first interval, that contains only $v_{i_0}$ and has out degree $0$). Then $|B_0|=1=  {q-1\choose 0}$. We conclude
	\[
	\left|B_{\overrightarrow{G}}(v,q)\right|=\sum_{s=0}^{w}|B_{s}|\le\sum_{s=0}^{w}{s+q-1 \choose s}\overset{(*)}{=}{w+q \choose w}~.
	\]
	To argue why equality $^{(*)}$ holds we will use a story. Denote by $A$ all the subsets of $[w+q]$ of size $w$. Clearly $|A|={w+q\choose w}$. In addition, for $s\in[0,w]$ note by $A_s$ all the subsets of $[w+q]$ where the maximum number not chosen is $q+s$. That is, every subset in $A_s$ have to contain all the number in $\{q+s+1,w+s+2,\dots,w+q\}$, and in addition $s$ number from $\{1,\dots,q+s-1\}$. That is $|A_s|={s+q-1\choose s}$. As $A=\cup_{s=0}^w A_s$, and all these sets are disjoint, it holds that ${w+q \choose w}=|A|=\sum_{s=0}^{w}|A_{s}|=\sum_{s=0}^{w}{s+q-1 \choose s}$, proving quality $^{(*)}$, and thus the induction and the lemma follows. 
\end{proof}
Next, we extend the supernode buffer property as a function of the distance in $\overrightarrow{G}_{\cC}$.
\begin{lemma}\label{lem:BufferExtended}
	Consider a vertex $v\in V$ that belongs to a supernode $v\in\eta\in\cC$, and consider an ancestor supernode $\eta'\in\cC$ of $\eta$ such that $\eta'\notin B_{\overrightarrow{G}_{\cC}}(\eta,2q+1)$ for some $q\ge0$. Then $d_{\dom(\eta')}(v, \eta') > (q+1)\cdot\gamma$.
\end{lemma}
\begin{proof}
The proof is by induction on $q$. The base case is when $q=0$, As $\eta'$ is an ancestor of $\eta$, $v\in \dom(\eta')$. However, as $\eta'\notin B_{\overrightarrow{G}_{\cC}}(\eta,1)$, $\eta$ and $\eta'$ are not adjacent in $G$. By the definition of buffered cop decomposition, it follows that  $d_{\dom(\eta')}(v, \eta') > \gamma$.

For the inductive step, 
let $P$ be the shortest path in $G[\dom(\eta')]$ from $v$ to $\eta'$. 
Let $\eta_1,\dots,\eta_k$ be the supernodes in $B_{\overrightarrow{G}_{\cC}}(\eta,1)$, and suppose that $\eta_k$ is the ancestor of $\eta,\eta_1,\dots,\eta_{k-1}$.
Let $u$ be the first vertex along $P$ out of $\dom(\eta_k)$. $u$ belong to some supernode $\eta_{u}$. Note that $\eta_{u}$ is an ancestor of $\eta_k$, and thus also of $\eta$. However, as $u\in\dom(\eta')$ it holds that $\eta'$ is an ancestor of $\eta_{u}$. 
Clearly, $\eta_{u}\notin B_{\overrightarrow{G}_{\cC}}(\eta,1)$. However, as there is an edge from $\dom(\eta_k)$ to $\eta_{u}$, it holds that $d_{\overrightarrow{G}_{\cC}}(\eta,\eta_{u})=2$, and hence $\eta'\notin B_{\overrightarrow{G}_{\cC}}(\eta_{u},2(q-1)+1)$. Using the induction hypothesis, it holds that $d_{\dom(\eta')}(u, \eta') > q\cdot\gamma$. Using the base of the induction, $d_{\dom(\eta_{u})}(v,u)>\gamma$. We conclude that $w(P)> q\cdot\gamma+\gamma=(q+1)\cdot\gamma$, as required.
\end{proof}

\subsection{Sparse Covers Construction}\label{subsec:SparseCoverCOnstruction}
This subsection is devoted to proving a the following meta theorem, which given a buffered cop decomposition constructs a sparse cover. As a corollary, we will obtain our main  \Cref{thm:buffered}.
\begin{theorem}[From buffered cop decomposition to sparse covers]\label{thm:ReductionBufferedToCovers}
	Consider a graph $G=(V,E,w)$ that admits a $(\Delta, \gamma,w)$-buffered cop decomposition where $\gamma\le\Delta$. Then for every $q\ge 1$, $G$ admits a strong $\left(4\cdot(\frac{2\Delta}{q\cdot\gamma}+1)~,~{w+2q-1 \choose w}\cdot4\cdot w\cdot(2+\frac{q\cdot\gamma}{\Delta})~,~2\cdot(2\Delta+q\cdot\gamma)\right)$-sparse partition cover.
\end{theorem}
%
%We are finally ready to prove the main theorem of the subsection, \Cref{thm:ReductionBufferedToCovers}.
\begin{proof}%[Proof of \Cref{thm:ReductionBufferedToCovers}]
Consider a $(\Delta, \gamma, w)$-buffered cop decomposition $\cT$ of a graph $G$.
%, parameter $\Delta>0$, and let $\cT$ be a buffered cop decomposition from \Cref{thm:buffered}.
Let $q$ be the input parameter, and let $k={w+2q-1\choose w}$.
Consider two supernodes $\eta,\eta'$, where $\eta'$ is an ancestor of $\eta$. We say that $\eta$ and $\eta'$ are \emph{nearby} if $d_{G_\cC}(\eta,\eta')\le 2q-1$. Note that by \Cref{lem:diballs}, each supernode $\eta\in\cC$ has at most $k$ nearby nodes (including itself).
We first partition the supernodes into sets $\cS_1,\cS_2,\dots,\cS_{k}$ such that no two nearby supernodes belong to the same set. This can be done greedily w.r.t the order induced by $\cT$.
Specifically, when we construct $\cS_1$ all the supernodes are active. The root $\eta$ joins $\cS_1$. We mark $\eta$ and all its nearby supernodes as inactive. 
Then all the maximal supernode w.r.t. the order induced by $\cT$ among the remaining active supernodes join $\cS_1$. We also mark the nearby supernodes of the newly joining supernodes as inactive. We continue with this process until no active supernodes remain.
In general, to construct $\cS_i$, we begin by marking all the supernodes in $\cup_{j=1}^{i-1}\cS_j$ as inactive, and continue with the same process. That is, all the maximal supernodes $\eta'$  (w.r.t. order induced by $\cT$ among the remaining active supernodes) join $\cS_i$, we then mark all their nearby supernodes as inactive, and so on. 

We argue that every supernode $\eta$ joins one of $\cS_1,\cS_2,\dots,\cS_{k}$.
Indeed, the only way $\eta$ might become inactive without joining $\cS_i$ is if one of its nearby ancestor supernodes joins $\cS_i$.
However, as $\eta$  has only $k-1$ nearby ancestor clusters (not counting itself), if $\eta$ was prevented from joining $\cS_1,\cS_2,\dots,\cS_{k-1}$, it will necessarily join $\cS_{k}$.
Consider $\cS_i$. For every supernode $\eta\in \cS_i$, let $\hat{\eta}=B_{G[\dom(\eta)]}(\eta,q\cdot\gamma)$ be the ball of radius $q\cdot\gamma$ around $\eta$ in $\dom(\eta)$. We argue the following:
\begin{lemma}\label{lem:CoverByCores}
	Consider the collection $\{\hat{\eta}\}_{\eta\in\cS_1},\dots,\{\hat{\eta}\}_{\eta\in\cS_{k}}$. It holds that:
	\begin{enumerate}
		\item For every $i\in[k]$, $\{\hat{\eta}\}_{\eta\in\cS_i}$ is a partial partition of $V$. That is for every $\eta_1,\eta_2\in\cS_i$, $\hat{\eta}_1\cap \hat{\eta}_2=\emptyset$.
		\item For every $i\in[k]$ and $\eta\in\cS_i$, $\hat{\eta}$ contains a skeleton $T_\eta$ which is a shortest path tree in $G[\hat{\eta}]$ with at most $w$ leafs.		
%		such that every vertex $v\in \hat{\eta}$ is at distance at most $\Delta+\gamma$ from $T_\eta$. Furthermore, $T_\eta$
%		is a shortest path tree in $G[\hat{\eta}]$ with at most $r-1$ leafs.
		\item For every vertex $v\in V$, there is some index $i\in[k ]$ and supernode $\eta\in\cS_i$ such that
		$d_{G[\hat{\eta}]}(v,T_\eta)\le \Delta+\frac{q\cdot\gamma}{2}$, and 
		 $B_G(v,\frac{q\cdot\gamma}{2})\subseteq\hat{\eta}$.
	\end{enumerate}
\end{lemma}
\begin{proof}
	We begin by proving the first point. Consider $\eta_1,\eta_2\in\cS_i$.
	If $\eta_1,\eta_2$ does not have ancestry relationship, then  $\dom(\eta_1),\dom(\eta_2)$ are disjoint, which implies  $\hat{\eta}_1\cap \hat{\eta}_2=\emptyset$. We thus can assume w.l.o.g. that $\eta_1$ is an ancestor of $\eta_2$. Suppose for the sake of contradiction that there is a vertex $v\in \hat{\eta}_1\cap \hat{\eta}_2$. It holds that $v\in \dom(\eta_1)$, and $d_{G[\dom(\eta_1)]}(\eta_1,v)\le q\cdot\gamma$. On the other hand, as $\eta_1,\eta_2\in\cS_i$ they are not nearby. In particular, $d_{G_\cC}(\eta_1,\eta_2)> 2q-1$ which implies $\eta_1\notin B_{\overrightarrow{G}}(\eta_2,2q-1)$. But by \Cref{lem:BufferExtended}, $d_{G[\dom(\eta_1)]}(\eta_1,v)\ge d_{G[\dom(\eta_1)]}(\eta_1,\eta_2)> q\cdot\gamma$, a contradiction.
	
	The second point follows directly from the constriction. 
	For the third point, let $\eta$ be the first supernode w.r.t. the order induced by $\cT$ that intersects the ball $B_G(v,\frac{q\cdot\gamma}{2})$. Suppose that $\eta\in\cS_i$. By the minimality of $\eta$, it holds that $B_G(v,\frac{q\cdot\gamma}{2})\subseteq \dom(\eta)$.
	In particular, as $\eta$ intersects $B_G(v,\frac{q\cdot\gamma}{2})$, by triangle inequality it follows that $d_{G[\hat{\eta}]}(v,T_\eta)\le \Delta+\frac{q\cdot\gamma}{2}$, and 
	$B_G(v,\frac{q\cdot\gamma}{2})\subseteq\hat{\eta}$.
\end{proof}

Next we argue that each enlarged supernode can be efficiently clustered.

\begin{lemma}\label{lem:SupernodePartitions}
	Consider an enlarged supernode $\hat{\eta}$. 
%	For every parameter $\eps\in(0,1)$, t
	There is set of partitions $\cC^{\hat{\eta}}_1,\dots,\cC^{\hat{\eta}}_{q}$ of size $s\le 4w\cdot(2+\frac{q\cdot\gamma}{\Delta})$ with strong diameter $2\cdot(2\Delta+q\cdot\gamma)$  such that for every vertex $v\in \hat{\eta}$ at distance at most $\Delta+\frac{q\cdot\gamma}{2}$ from $T_\eta$, the ball $B_{G[\hat{\eta}]}(v,\frac{q\cdot\gamma}{2})$ is fully contained in some cluster in some partition.
\end{lemma}
\begin{proof}
	Let $N\subseteq T_\eta$ be an $\Delta$-net. That is, a set of points in $T_\eta$ at pairwise distance greater than $\Delta$, and such that for every point $v\in T_\eta$, there is some net point $u\in N$ at distance at most $d_{G[T_\eta]}(u,v)\le\Delta$. Such a set can be constructed greedily.
	We say that a pair of net points $u,v\in N$ is nearby if $d_{G[\hat{\eta}]}(v,u)\le2\cdot \left(2\Delta+q\cdot\gamma\right)$. 
	\begin{claim}
		Every net point $v\in N$ has at most $s=4w\cdot(2+\frac{q\cdot\gamma}{\Delta})$
		 nearby net points (including $v$).
	\end{claim}
	\begin{proof}
		$T_\eta$ is a shortest path tree with at most $w$ leafs. In particular, $T_\eta$ consist of the shortest paths $P_1,\dots,P_k$ for $k\le w$.
		Fix some $P_i$. Let $u_1,u_2,\dots,u_t$ be the net points along $P_i$ which are nearby $v$, sorted w.r.t. their position along $P_i$. As $P_i$ is a shortest path, and the pairwise distance between every pair of consecutive net points is at least $\Delta$, we have that $d_{G[\hat{\eta}]}(u_1,u_t)>(t-1)\cdot\Delta$. 
		On the other hand, as both $u_1,u_t$ are nearby $v$, by the triangle inequality it holds that 
		$d_{G[\hat{\eta}]}(u_{1},u_{t})\le d_{G[\hat{\eta}]}(u_{1},v)+d_{G[\hat{\eta}]}(v,u_{t})\le4\cdot(2\Delta+q\cdot\gamma)$. It follows that $(t-1)\cdot\Delta<4\cdot(2\Delta+q\cdot\gamma)$
		and thus $t\le4\cdot(2+\frac{q\cdot\gamma}{\Delta})$. The claim now follows as there are at most $w$ such paths. 		
	\end{proof}
	Partition $N$ into subsets $N_1,\dots,N_s$ such that no two nearby points will belong to the same set $N_i$ (this can be done in a greedy manner).
	For every subset $N_i$, consider the partition with the clusters $\left\{B_{G[\hat{\eta}]}(v,2\Delta+q\cdot\gamma)\mid v\in N_i\right\}$, and the rest of $\hat{\eta}$ vertices being singletons.
	Note that this is indeed a partition as $N_i$ does not contain nearby net points and hence all these balls are disjoint. Clearly the strong diameter of each cluster is at most $2\cdot (2\Delta+q\cdot\gamma)$. 
	Finally, consider a vertex $v\in \hat{\eta}$ such that there is $x\in T_\eta$ with $d_{G[\hat{\eta}]}(v,x)\le\Delta+\frac{q\cdot\gamma}{2}$. Let $y\in N$ be the closest net point to $x$. 
	By the triangle inequality $d_{G[\hat{\eta}]}(v,y)\le d_{G[\hat{\eta}]}(v,x)+d_{G[\hat{\eta}]}(v,x)\le 2\Delta+\frac{q\cdot\gamma}{2}$. 
	Let $i$ be an index such that $y\in N_i$.
	It follows that the ball 
	$B_{G[\hat{\eta}]}(y,2\Delta+q\cdot\gamma)$, which belongs to the respective partition of $N_i$ contains $B_{G[\hat{\eta}]}(v,\frac{q\cdot\gamma}{2})$ as required.
\end{proof}

Our final set of partitions is constructed naturally. 
We first use \Cref{lem:CoverByCores} to obtain the collections $\{\hat{\eta}\}_{\eta\in\cS_1},\dots,\{\hat{\eta}\}_{\eta\in\cS_{k}}$ of enlarged supernodes.
Then for every $i\in[k]$, and every $\hat{\eta}\in\cS_i$, we use \Cref{lem:SupernodePartitions} to obtain partitions $\cC^{\hat{\eta}}_1,\dots,\cC^{\hat{\eta}}_{s}$.
Finally, for every $i\in[k]$, and $j\in [s]$, we create the partition $\left\{\cC_j^{\hat{\eta}}\right\}_{\eta\in\cS_i}$ of $G$, and add all the non clustered vertices as singletons.
Note that we obtained $k\cdot s={w+2q-1 \choose w}\cdot4\cdot w\cdot(2+\frac{q\cdot\gamma}{\Delta})$
 partitions. Clearly, each cluster in every partition has strong diameter at most $2\cdot(2\Delta+q\cdot\gamma)$. Finally, we prove the cover property. Consider a vertex $v\in V$. 
By \Cref{lem:CoverByCores} there is some index $i\in[k]$ and supernode $\eta\in\cS_i$ such that
$d_{G[\hat{\eta}]}(v,T_\eta)\le \Delta+\frac{q\cdot\gamma}{2}$, and 
$B_G(v,\frac{q\cdot\gamma}{2})\subseteq\hat{\eta}$.
By \Cref{lem:SupernodePartitions},  the ball $B_{G[\hat{\eta}]}(v,\frac{q\cdot\gamma}{2})=B_{G}(v,\frac{q\cdot\gamma}{2})$ is fully contained in some cluster in some partition. Thus our padding parameter is $\frac{2\cdot(2\Delta+q\cdot\gamma)}{\frac{q\cdot\gamma}{2}}=4\cdot(\frac{2\Delta}{q\cdot\gamma}+1)$, as required.
\end{proof}

We are now ready to prove our main \Cref{thm:MinorFreeCover} (restated for convenience).
%Our main \Cref{thm:MinorFreeCover} is then a corollary of \Cref{thm:ReductionBufferedToCovers} and \Cref{thm:buffered}.
\MinorCover*
\begin{proof}
	\sloppy According to \cite{CCLMST24} (\Cref{thm:buffered}), every $K_r$ minor free graph $G$ admits $(\Delta, \Delta/r, r-1)$-buffered cop decomposition for every $\Delta>0$. Using \Cref{thm:ReductionBufferedToCovers} it follows that $G$ admits strong $\left(4\cdot(\frac{2r}{q}+1)~,~{r+2q-2 \choose r-1}\cdot4\cdot r\cdot(2+\frac{q}{r})~,~2\cdot(2+\frac{q}{r})\cdot\Delta\right)$-sparse partition cover, for arbitrary $\Delta>0$. It follows that every $K_r$ minor free graph admits strong $\left(4\cdot(\frac{2r}{q}+1)~,~{r+2q-2 \choose r-1}\cdot4\cdot r\cdot(2+\frac{q}{r})\right)$-\SPCS.
	
	Fixing $q=1$ we obtain that every $K_r$ minor free graph admits strong $\left(4\cdot(2r+1)~,~{r \choose r-1}\cdot4\cdot r\cdot(2+\frac{1}{r})\right)=\left(O(r),O(r^{2})\right)$-\SPCS.
%	Similarly, if we fix $q=r$ we obtain that $G$ admits a strong  $\left(O(1),O(1)^{r}\right)$-\SPCS.
	Next, fix $q=\frac{8r}{\eps}$, using the identity ${n\choose k}\le\left(\frac{n\cdot e}{k}\right)^k$, we have ${r+\frac{16r}{\eps}-2 \choose r-1}\cdot4\cdot r\cdot(2+\frac{8}{\eps})\le\left(\frac{(r+\frac{16r}{\eps}-2)\cdot e}{r-1}\right)^{r-1}\cdot O(\frac{r}{\eps})=O(\frac{1}{\eps})^{r}$, and thus $G$ admits a strong  $\left(4+\eps,O(\frac{1}{\eps})^{r}\right)$-\SPCS.	The theorem follows by scaling $\eps$ accordingly.
\end{proof}

\section{Metric Embedding into $\ell_\infty$ from Sparse Partition Cover}

The main result of this section is a meta theorem (\Cref{thm:SPCStoEmbeddingNoAspect}) that transforms a \SPCS into a low distortion/dimension metric embedding into $\ell_\infty$. 
This reduction is applicable for graph families which are closed under re-weighting (abbreviated \CURW).
A graph family $\cF$ is \CURW  if for every  weighted graph $G=(V,E,w)\in{\cal F}$ in the family,
it holds that for every weight function $w':E\rightarrow\mathbb{R}_{\ge0}$, the reweighed graph
$G'=(V,E,w')\in{\cal F}$ is also in the family. 
Many graph families are CURW,
some examples being general graphs, planar graphs, graphs with bounded treewidth\textbackslash pathwidth \textbackslash genus, and most importantly for us, $H$ minor free graphs.
There are also well studied families which are not CURW, such as graphs
with bounded doubling dimension, or bounded highway dimension. 

\begin{restatable}[From \SPCS to Metric embedding]{theorem}{CoverToEmbedding}
	\label{thm:SPCStoEmbeddingNoAspect}
		Consider a CURW graph family ${\cal F}$, such that every $G\in{\cal F}$ admits a $(\beta,\tau)$-\SPCS. Then for every $\eps>0$, every $n$-point graph $G=(V,E,w)\in{\cal F}$ admits an efficiently computable embedding $f:V\rightarrow \ell_\infty^k$ with distortion $(1+\eps)\cdot2\beta$ and dimension $k=O\left(\frac{\tau}{\eps}\cdot\log\frac{\beta}{\eps}\cdot\log(\frac{n\cdot\beta}{\eps})\right)$.
\end{restatable}

Using our \Cref{thm:MinorFreeCover} with $\eps=1$, \Cref{cor:EmbeddingMinor} follows. We restate it for convenience:
\CorEmbeddingMinorFree*

%\begin{corollary}\label{cor:EmbeddingMinor2}
%	Every $n$ vertex $K_r$-minor free graph $G$ can be embedded into $\ell_\infty^{O(r^{2}\cdot\log r\cdot\log n)}$ with distortion $O(r)$. Alternatively,  $G$ can be embedded into $\ell_\infty^{O(2^{O(r)}\cdot\log n)}$ with distortion $O(1)$.
%\end{corollary}
One can also use our $\left(4+\eps,O(\frac{1}{\eps})^{r}\right)$-\SPCS to obtain embedding into $\ell_\infty^{O(\frac{1}{\eps})^{r+1}\cdot\log\frac{1}{\eps}\cdot\log\frac{n}{\eps}}$ with distortion $8+\eps$.
However, for the case of minor-free graphs, we open the black box of \Cref{thm:SPCStoEmbeddingNoAspect} and obtain an improved embedding with distortion $3+\eps$ (see \Cref{thm:EmbeddingMinor3}) . 
%\begin{restatable}[]{theorem}{EmbeddingMinorFree2}
%	\label{thm:EmbeddingMinor32}
%	For every $\eps\in(0,\frac12)$, every $n$ vertex $K_r$-minor free graph $G$ can be embedded into $\ell_\infty^{O(\frac{1}{\eps})^{r+1}\cdot\log\frac{1}{\eps}\cdot\log\frac{n}{\eps}}$ with distortion $3+\eps$.
%\end{restatable}
%
%We show that given a CURW graph family which admits a \SPCS, one can efficiently construct a low dimensional embedding into $\ell_\infty$. Interestingly, this holds also for non CURW graph families, however the dimension will be $O\left(\frac{\tau}{\eps}\cdot\log\frac{\beta}{\eps}\cdot\log(n\Phi)\right)$ (here $\Phi$ is the aspect ratio\atodo{define aspect ratio}), see \Cref{subsec:ImproveEmbeddingMinorFree}.

The plan for the section is as follows: first in \Cref{subsec:PrefixFree} we construct a prefix free code, which will later be used in the construction of the embedding. Then, in \Cref{subsec:CoversToEmbeddingLargeAspect} we construct a reduction from \SPCS into metric embedding (see \Cref{lem:fromSparseCoverToEmbedding}). Note that there is no assumption of a \CURW graph family here. The resulting dimension is $O\left(\frac{\tau}{\eps}\cdot\log\frac{\beta}{\eps}\cdot\log(n\Phi)\right)$, where $\Phi=\frac{\max_{u,v}d_G(u,v)}{\min_{u,v}d_G(u,v)}$ is the aspect ratio.
Then in \Cref{subsec:RemoveAspect} we use \Cref{lem:fromSparseCoverToEmbedding} to obtain our main \Cref{thm:SPCStoEmbeddingNoAspect}.
Finally, in \Cref{subsec:ImproveEmbeddingMinorFree} we improve the result for minor free graph to obtain distortion $3+\eps$.

\subsection{Prefix Free Codes}\label{subsec:PrefixFree}
This sub-section is devoted to the construction of prefix free codes. $\{ \pm1\} ^{*}$ stands for all finite sequences of $\{ \pm1\}$.
\begin{definition}
	A code for elements $X$ is a function $h:X\rightarrow\left\{\pm1\right\} ^{*}$.
	The code is prefix free if for every $x,y\in X$, $h(x)$ and $h(y)$
	are not the prefixes of each other. That is, there is no $L\in \{ \pm1\} ^{*}$ such that either $h(x)\circ L= h(y)$, or $h(y)\circ L= h(x)$.
\end{definition}
In the following lemma we construct a prefix free code where the length of the code word of $x$ is bounded by $2\cdot\left\lceil \log\frac{\mu(X)}{\mu(x)}\right\rceil$. This statement is not new, and even better prefix free codes are known (see e.g. \cite{KN76}). Nevertheless, we provide a proof of  \Cref{lem:Huffman} in \Cref{app:Codes} for the sake of completeness.
\begin{restatable}[]{lemma}{Hufmann}\label{lem:Huffman}
	Consider a set $X$ with a weight function $\mu:X\rightarrow\mathbb{R}_{>0}$.
	Denote $\mu(X)=\sum_{x\in X}\mu(X)$. Then there is a prefix free
	code $h$, such that for every $x\in X$, $|h(x)|\le2\cdot\left\lceil \log\frac{\mu(X)}{\mu(x)}\right\rceil$.
\end{restatable}
%\begin{lemma}\label{lem:Huffman}
%	Consider a set $X$ with a weight function $\mu:X\rightarrow\mathbb{R}_{>0}$.
%	Denote $\mu(X)=\sum_{x\in X}\mu(X)$. Then there is a prefix free
%	code $h$, such that for every $x\in X$, $|h(x)|\le2\cdot\left\lceil \log\frac{\mu(X)}{\mu(x)}\right\rceil $.
%\end{lemma}

\subsection{Embedding into $\ell_\infty$ from sparse partition cover scheme}\label{subsec:CoversToEmbeddingLargeAspect}
In this subsection, given a \SPCS, we construct a metric embedding into $\ell_\infty$, where the dimension depans on the aspect ratio (\Cref{lem:fromSparseCoverToEmbedding}). 
We begin this subsection by showing that we can turn a collection of unrelated partitions in different scales into a laminar partition (\Cref{lem:FromCoverToLaminar}).
This statement  was implicitly proved in \cite{KLMN04}, as well as in other constructions. The following lemma is an adaptation from \cite{FL22}. We provide a proof for the sake of completeness.
\begin{lemma}\label{lem:FromCoverToLaminar}
	Consider a finite metric space $(X,d_X)$ that admits a $(\beta,\tau)$-sparse partition cover scheme. Then for every $a>0$, and $\eps\in(0,1)$ there is a collection of hierarchical partitions of $X$: $\{\cP_q^i\}_{q\in\Z}$, $i\in[\tau]$ with the following properties:
	\begin{enumerate}
		\item For each $j$, $\{\cP_i^j\}_{i\in\Z}$ is a hierarchical partition. That is $\forall i,j$, $\cP_i^j$ is a refinement of $\cP_{i+1}^{j}$.
		\item Fix $\Delta_{i}=a\cdot(\frac{4\beta}{\epsilon})^{i}$. For every $j$ and $i\in\Z$, $\cP_i^j$ has diameter $(1+\eps)\cdot\Delta_i$.
		\item For every $x\in X$ and $i\in\Z$, there is some $j\in[\tau]$ such that the ball $B_X(x,\frac{\Delta_i}{\beta(1+\eps)})$ is fully contained in some cluster of $\cP_i^j$.
	\end{enumerate}
\end{lemma}
\begin{proof}
	Assume w.l.o.g. that the minimal pairwise distance in $X$ is $1$ (otherwise one can scale accordingly), while the maximal  pairwise distance is $\Phi$.	
	Set $I = \lceil \log_{\nicefrac{4\beta}{\epsilon}}\nicefrac\Phi a\rceil$.
	For $i\in[0,I]$, let
	$\mathbb{P}_i=\{\mathcal{P}_{i}^{1},\dots,\mathcal{P}_{i}^{\tau}\}$ be a
	$(\beta,\tau,\Delta_{i})$-sparse partition cover (we assume that $\cP_i$ has exactly $\tau$ partitions, we can enforce this assumption by duplicating partitions if necessary). 
	Fix some $j$. For $i<0$, let $\mathcal{P}^{j}_{i}$ be the partition where each vertex is a singleton. For $i>I$ let 
	$\mathcal{P}^{j}_{i}$ be the trivial partition with a single cluster $\{X\}$.	
	Consider $\{\mathcal{P}^{j}_{i}\}_{i\ge 0}^{I}$.
	We will inductively define a new set of partitions $\{\tilde{\mathcal{P}}^{j}_{i}\}_{i\ge 0}^{I}$, enforcing it to be a laminar system, while still closely resembling the original partitions. 
%	The basic idea of doing this is to produce a tree of partitions where the lower level is a refinement of the higher level, and we do so by grouping a cluster at a lower level to one of the clusters at a higher level intersecting it. 
	
	Levels bellow $0$ and above $I$ stay as-is. 
	Inductively, for any $i\geq 0$, after constructing $\tilde{\mathcal{P}}^{j}_{i-1}$ from $\mathcal{P}^{j}_{i-1}$, we will construct $\tilde{\mathcal{P}}^{j}_{i}$
	from $\mathcal{P}^{j}_{i}$ using $\tilde{\mathcal{P}}^{j}_{i-1}$.
	Let $\mathcal{P}_{i}^{j}=\left\{ C_{1},\dots,C_{\phi}\right\}$ be the clusters in the partition $\mathcal{P}_{i}^{j}$. For each $q\in[1,\phi]$, let $Y_{q}=X\setminus\cup_{a<q}\tilde{C}_{a}$ be the set of unclustered points (w.r.t. level $i$, before iteration $q$). 
	Let	$C'_{q}=C_{q}\cap Y_{q}$ be the cluster $C_q$ restricted to vertices in $Y_q$, and let $S_{C'_{q}}=\left\{ C\in\tilde{\mathcal{P}}_{i-1}^{j}\mid C\cap C'_{q}\ne\emptyset\right\}$ be  the set of new level-$(i-1)$ clusters with non empty intersection with $C'_{q}$.
	We set the new cluster $\tilde{C}_{q}=\cup S_{C'_{q}}$ to be the union of these clusters, and continue iteratively.
%	See \Cref{fig:Laminar} for illustration.
	Clearly, $\tilde{\mathcal{P}}_{i-1}^{j}$
	is a refinement of $\tilde{\mathcal{P}}_{i}^{j}$. We conclude that
	$\left\{ \tilde{\mathcal{P}}_{i}^{j}\right\} _{i\ge-1}$ is a laminar hierarchical set of partitions that refine each other.

	We next argue by induction that $\tilde{\mathcal{P}}_{i}^{j}$ has diameter $(1+\eps)\cdot\Delta_{i}$.
Consider $\tilde{C}_{q}\in\tilde{\mathcal{P}}_{i}^{j}$. It consists
of $C'_{q}\subseteq C_{q}\in\mathcal{P}_{i}^{j}$ and of clusters in $\tilde{\mathcal{P}}_{i-1}^{j}$
intersecting $C'_{q}$. 
As the diameter of $C'_{q}$ is bounded by $\diam(C_{q})\le \Delta_i$, and by the induction hypothesis, the diameter of each cluster $C\in \tilde{\mathcal{P}}_{i-1}^{j}$ is bounded by $(1+\eps)\cdot\Delta_{i-1}$, we conclude that the diameter of $\tilde{C}_{q}$ is bounded by
\[
\Delta_{i}+2\cdot(1+\epsilon)\Delta_{i-1}=\Delta_{i}\cdot\left(1+\frac{2(1+\epsilon)}{4\beta}\cdot\epsilon\right)\le(1+\eps)\cdot\Delta_{i}~,
\]
since $\beta \geq 1$ and $\eps < 1$.

Finally, we argue that every ball is fully contained in some cluster.
Fix, $x\in X$ and $i\in \Z$.
If $i<0$ or $i>I$ then there is nothing to prove. We will thus assume $i\in[0,I]$.
Set
\[
R=\frac{\Delta_{i}}{\beta}-(1+\eps)\cdot\Delta_{i-1}=\frac{\Delta_{i}}{\beta}\cdot\left(1-(1+\eps)\cdot\frac{\epsilon}{4}\right)\ge\frac{\Delta_{i}}{(1+\eps)\cdot\beta}~.
\]
%It will be enough to show that the ball $B_X(x,R)$ is full contain in some cluster in one of $\tilde{P}^1_i,\dots,\tilde{P}^\tau_i$.
By construction, there is some index $j$ such that the ball $B_X(x,\frac{\Delta_i}{\beta})$ is fully contained in a single cluster $C_q$ in the partition $\cP_i^j$. 
%Set $R=\frac{\Delta_i}{\beta}-(1+\eps)\cdot \Delta_{i-1}$. 
It will be enough to show that the ball $B_X(x,R)$ is fully contained in the cluster $\tilde{C}_q$ of $\tilde{\cP}_i^j$. 
Suppose for the sake of contradiction otherwise, then there have to be a vertex $y\in B_X(x,R)\setminus Y_q$. In particular, $y$ belongs to a cluster $Q\in\tilde{\cP}_j^{i-1}$ that intersected $C_{q'}\in \cP_i^j$ for $q'<q$. Let $z\in C_{q'}\cap Q$. As the maximum diameter of $\tilde{\cP}_j^{i-1}$ cluster is $(1+\eps)\cdot\Delta_{i-1}$, it follows that 
$$d_X(z,x)\le d_X(z,y)+d_X(y,x)\le (1+\eps)\cdot\Delta_{i-1}+R=\frac{\Delta_i}{\beta}~.$$
But this implies that $z\in C_q$ and not $C_{q'}$, a contradiction.
%We conclude that the ball of radius 
%\[
%R=\frac{\Delta_{i}}{\beta}-(1+\eps)\cdot\Delta_{i-1}=\frac{\Delta_{i}}{\beta}\cdot\left(1-(1+\eps)\cdot\frac{\epsilon}{4}\right)\ge\frac{\Delta_{i}}{(1+\eps)\cdot\beta}~,
%\]
%around $x$ is fully contained in $\tilde{C}_q$, as required.

\end{proof}

Consider a partition $\cP$, and a point $x\in X$. Let $C_x\in\cP$ be the cluster containing $x$. We denote by $\partial_{\cP}(x)=d_X(x,X\setminus C_x)$ the distance between $x$ and the boundary of it's cluster $C_x$.
The following embedding  will be the core of our construction:
\begin{lemma}\label{lem:singleHierarchy}
	Consider a finite metric space $(X,d_X)$, and a hierarchy of laminar partitions $\cP_1,\dots,\cP_k$ of $X$ (that is $\cP_i$ is a refinement of $\cP_{i+1})$. Then there is an embedding $f:X\rightarrow\ell_\infty^{2\log |X|+2k}$ such that:
	\begin{enumerate}
		\item Lipschitz: $\forall x,y\in X,~\|f(x)-f(y)\|_\infty\le 2\cdot d_X(x,y)$.
		\item For every pair of points $x,y\in X$ and partition $\cP_i$ where $x,y$ belong to different clusters, it holds that $\|f(x)-f(y)\|_\infty\ge \partial_{\cP_i}(x)+\partial_{\cP_i}(y)$.
%		\max\{\partial_{\cP_i}(x),\partial_{\cP_i}(y)\}$.
	\end{enumerate}
\end{lemma}
\begin{proof}
	The construction and proof is by induction on $|X|$, the number of metric points. The base case where $|X|=1$ is trivial, as the empty embedding will do.
	Suppose that the partition $\cP_k$ consists of the clusters $X_1,X_2,\dots,X_m$. If $m=1$ (that is the trivial partition), then we can simply ignore this partition (as we don't guarantee anything w.r.t. this partition).
	If all the partitions in the hierarchy are trivial, then again there is nothing to prove (as no pair is ever separated) and the empty embedding will do. We thus can assume that $m\ge2$.

	Using the induction hypothesis, we construct functions $f_{j}:X_{j}\rightarrow\ell_{\infty}^{2\log|X_{j}|+2(k-1)}$
	fulfilling the conditions of the lemma. Let $Y$ be a set with the
	elements being $X_{1},\dots,X_{m}$ and weight function $\mu(X_{j})=|X_{j}|$. Note that $\mu(Y)=|X|$. Using \Cref{lem:Huffman} we obtain a
	prefix free code $h:\{X_1,\dots,X_m\}\rightarrow\{\pm1\}^*$ such that for every $j\in[1,m]$, $|h(X_{j})|=2\cdot\left\lceil \log\frac{\mu(Y)}{\mu(X_{j})}\right\rceil =2\cdot\left\lceil \log\frac{|X|}{|X_{j}|}\right\rceil $
	(we can obtain equality instead of weak inequality by padding with
	arbitrary $\pm1$). We define a new embedding $f:X\rightarrow\ell_{\infty}$,
	such that for a point $x\in X_{j}$, for every $q\le2\cdot\left\lceil \log\frac{\mu(Y)}{\mu(X_{j})}\right\rceil$, the $q$`th coordinate equal to
	$(f(x))_{q}=(h(X_{j}))_{q}\cdot \partial_{\cP_k}(x)$.
	The rest of the coordinates will be just a concatenation of $f_{j}(x)$.
	For every point $x\in X_j$, the total number of coordinates is bounded by
	\[
	2\log|X_{j}|+2(k-1)+2\cdot\left\lceil \log\frac{|X|}{|X_{j}|}\right\rceil ~~\le~~2\log|X_{j}|+2k+2\cdot\log\frac{|X|}{|X_{j}|}~~=~~2k+2\log |X|~~.
	\]
	We can pad with $0$ coordinates to get exactly $2k+2\log |X|$ coordinates.
	
	Next we prove that our embeddings is expanding by at most a factor of $2$.
	The proof follows the construction, and hence it is also by induction on $|X|$. Consider a pair $x,y\in X$. Suppose first that $x,y$ belong to the same cluster $X_j$ in the partition $\cP_k$.
	For every coordinate $q$ among the first $2\cdot\left\lceil \log\frac{|X|}{|X_{j}|}\right\rceil$ coordinates, by the triangle inequality it holds that 
	\begin{align}
		\left|(f(x))_{q}-(f(y))_{q}\right| & =\left|h_{q}(X_{j})\cdot\partial_{\cP_{k}}(x)-h_{q}(X_{j})\cdot\partial_{\cP_{k}}(y)\right|\nonumber\\
		& =\left|d_{X}(x,X\setminus X_{j})-d_{X}(y,X\setminus X_{j})\right|\le d_{X}(x,y)~.\label{eq:Lipshitz1}
	\end{align}
	The rest of the coordinates are a concatenation of $f_j(x),f_j(y)$ and the upper bound follows by the induction hypothesis.
	
	Next, suppose that $x\in X_j$, $y\in X_{j'}$ for $j\ne j'$. Consider a coordinate $q$. Note that $(f(x))_q$ is either a $0$, or equals to $\sigma_x\cdot \partial_{\cP_{k'}}(x)=\sigma_x\cdot d_{X}(x,X\setminus Y_{j})$ where $\sigma_x\in\{\pm 1\}$, $k'\in[1,k]$, and $Y_j$ is the cluster containing $x$ in $\cP_{k'}$. In particular, $Y_j\subset X_j$. Similarly $(f(y))_q$ is either a $0$ or $\sigma_y\cdot d_{X}(x,X\setminus Y_{j'})$ for $\sigma_y\in\{\pm 1\}$ and $Y_{j'}\subseteq X_{j'}$. It holds that 
	\begin{align}
		\left|(f(x))_{q}-(f(y))_{q}\right| & \le\left|(f(x))_{q}\right|+\left|(f(y))_{q}\right|\le\left|d_{X}(x,X\setminus Y_{j})\right|+\left|d_{X}(y,X\setminus Y_{j'})\right|\nonumber\\
		& \le\left|d_{X}(x,y)\right|+\left|d_{X}(y,x)\right|=2\cdot d_{X}(y,x)~,\label{eq:Lipshitz2}
	\end{align}
	where the third inequality holds as $y\notin Y_j$ and $x\notin Y_{j'}$.
	
	It remains to prove the lower bound. Consider a pair $x,y\in X$ that belong to different clusters in $\cP_i$. Let $X_{j},X_{j'}\in\cP_{i}$ be the clusters containing $x$ and $y$ respectively. 
	Let $i'\in[i,k]$ be the maximal index such that $x,y$ belong to different clusters in $\cP_{i'}$. Let $Y_{j},Y_{j'}\in\cP_{i'}$ be the clusters containing $x$ and $y$ respectively. As $i\le i'$, $\cP_i$ is a refinement of $\cP_{i'}$. In particular,  $X_j\subseteq Y_j$ and $X_{j'}\subseteq Y_{j'}$.
	As $x$ and $y$ belong to the same clusters in the partitions $\cP_{j'+1},\dots,\cP_{k}$, their first coordinates of $f(x),f(y)$ correspond to these partitions and are ``aligned''. Suppose that all these embedding fill the first $s$ coordinates.	
	Then we have the partition w.r.t.  $\cP_{i'}$. In particular, we used a prefix free code $h$, and hence there is some $q$ such that both $h_q(Y_j),h_q(Y_{j'})$ are defined and different. We conclude
	\begin{align*}
		\left\Vert f(x)-f(y)\right\Vert _{\infty} & \ge\left|(f(x))_{s+q}-(f(y))_{s+q}\right|\\
		& =\left|h_{q}(Y_{j})\cdot\partial_{\cP_{j'}}(x)-h_{q}(Y_{j'})\cdot\partial_{\cP_{j'}}(y)\right|\\
		& =\left|\partial_{\cP_{j'}}(x)\right|+\left|\partial_{\cP_{j'}}(y)\right|~~\ge~~\left|\partial_{\cP_{j}}(x)\right|+\left|\partial_{\cP_{j}}(y)\right|~,
	\end{align*}
	where the last inequality holds as $\partial_{\cP_{j'}}(x)=d_{X}(x,X\setminus Y_{j})\ge d_{X}(x,X\setminus X_{j})=\partial_{\cP_{j}}(x)$
	because $X_j\subseteq Y_j$. Similarly for $y$.
	
\end{proof}

\begin{lemma}\label{lem:fromSparseCoverToEmbedding}
	Consider an $n$-point metric space $(X,d_X)$ that admits a $(\beta,\tau)$-sparse partition cover scheme, and such that all the pairwise distances in $X$ are between $1$ and $\Phi$. Then for every  $\eps\in(0,1)$, there is an (efficiently computable) embedding $f:X\rightarrow\ell_\infty^k$ with distortion $(1+\eps)\cdot2\beta$, for $k=O\left(\frac{\tau}{\eps}\cdot\log\frac{\beta}{\eps}\cdot\log(n\Phi)\right)$.
\end{lemma}
\begin{proof}
	We will construct embedding with distortion $(1+O(\eps))\cdot\beta$, afterwards one can obtain the lemma by scaling $\eps$ accordingly.
	For every $q\in\left\{1,2,\dots,\left\lceil \log_{1+\eps}\frac{4\beta}{\eps} \right\rceil\right\}$ let $a_q=(1+\eps)^q$.
	Using \Cref{lem:FromCoverToLaminar} with parameters $a_q$ and $\eps$, construct hierarchical partitions
	 $\{\cP_i^{1,q}\}_{i\in\Z},\dots,\{\cP_i^{\tau,q}\}_{i\in\Z}$.
	 Fix $I = \lceil \log_{\nicefrac{4\beta}{\epsilon}}\Phi \rceil$.
	For every $q\in\left\lceil \log_{1+\eps}\frac{4\beta}{\eps} \right\rceil$ and $j\in [\tau]$, consider the hierarchical partition  $\{\cP_i^{j,q}\}_{i=-1}^{I}$.
	Using \Cref{lem:singleHierarchy}, we construct an embedding $f_{j,q}:X\rightarrow \ell_\infty^{2\log n+2(I+2)}$.
	Our final embedding is simply a concatenation of all this embeddings: $\circ_{j\in[\tau],q\in[\lceil\log_{1+\eps}\frac{4\beta}{\eps}\rceil]}f_{j,q}$. 
	The overall dimension is $\tau\cdot\lceil\log_{1+\eps}\frac{4\beta}{\eps}\rceil\cdot\left(2\log n+2(I+1)\right)=O\left(\frac{\tau}{\eps}\cdot\log\frac{\beta}{\eps}\cdot\log(n\Phi)\right)$.

	Next we bound the distortion.
	Consider a pair $x,y\in X$.
	Let $q\in[0,\lfloor\log_{1+\epsilon}\frac{4\beta}{\epsilon}\rfloor]$, and $i\in[-1,I]$ be indices such that $(1+\epsilon)^{q}(\frac{4\beta}{\epsilon})^i< \frac{ d_{X}(x,y)}{1+\eps}\le(1+\epsilon)^{q+1}(\frac{4\beta}{\epsilon})^i$. 
	Fix $\Delta_i=a_q\cdot (\frac{4\beta}{\eps})^i$,
	and note that $(1+\eps)\cdot\Delta_i=(1+\eps)^{q+1}\cdot (\frac{4\beta}{\eps})^i<d_X(x,y)$. There is some $j\in[\tau]$ such that in the hierarchical partition $\{\cP_{i'}^{j,q}\}_{i'=-1}^I$, the ball  $B_X(x,\frac{\Delta_i}{\beta(1+\eps)})$ is fully contained in some cluster of $\cP_i^{j,q}$. In particular, $\partial_{\cP^{j,q}_i}(x)>\frac{\Delta_i}{\beta(1+\eps)}$.
	According to \Cref{lem:FromCoverToLaminar}, $\cP_i^{j,q}$ is a $(1+\eps)\cdot\Delta_i$ bounded partition. It follows that $x$ and $y$ have to belong to different clusters in $\cP_i^{j,q}$. Using \Cref{lem:singleHierarchy} we conclude
	\begin{align}
		\left\Vert f(x)-f(y)\right\Vert _{\infty} & \ge\left\Vert f_{j,q}(x)-f_{j,q}(y)\right\Vert _{\infty}\ge\partial_{\cP_{i}^{j,q}}(x)+\partial_{\cP_{i}^{j,q}}(y)\nonumber\\
		& \ge\frac{\Delta_{i}}{\beta(1+\eps)}\ge\frac{d_{X}(x,y)}{\beta(1+\eps)^{3}}=\frac{d_{X}(x,y)}{\beta(1+O(\eps))}~.\label{eq:contraction}
	\end{align}
	Furthermore, as all the embeddings created by \Cref{lem:singleHierarchy} are $2$-Lipschitz, so is $f$. The lemma now follows.
	
\end{proof}

\subsection{Removing the dependence on the aspect ratio}\label{subsec:RemoveAspect}
% Preview source code from paragraph 7 to 17
In this section we prove \Cref{thm:SPCStoEmbeddingNoAspect}.
The following is the main lemma of the subsection, where we show that in a metric embeddings from a CURW family to $\ell_{\infty}$, one can remove the dependence on the aspect ratio. This is a general phenomena, independent of \SPCS.
\begin{lemma}\label{lem:CURWremoveAspectRatio}
	Consider a CURW graph family ${\cal F}$, and suppose that every $n$
	- point graph $G\in{\cal F}$ with aspect ratio $\Phi$ can be embedded
	into $\ell_{\infty}^{k}$ with expansion $\rho$, contraction $\beta$,
	and dimension $k=\varphi(n,\Phi,\rho,\beta)$, where $\varphi:\mathbb{N}\times\mathbb{R}_{\ge1}^{3}\rightarrow\mathbb{N}$
	is some coordinate-wise monotone function. Then for every $\epsilon\in(0,\frac{1}{2})$,
	every $n$ - point graph $G\in{\cal F}$ can be embedded into $\ell_{\infty}^{k'}$
	with expansion $(1+\delta)\cdot\rho$, contraction $(1+\epsilon)\cdot\beta$,
	and dimension $3\cdot\varphi(n,\left(\frac{8\cdot\rho\cdot\beta}{\epsilon}\right)^{2}\cdot n^{3},\rho,\beta)$.
\end{lemma}
Combining \Cref{lem:fromSparseCoverToEmbedding,lem:CURWremoveAspectRatio} we get our main meta \Cref{thm:SPCStoEmbeddingNoAspect}.
%, restated for convenience.
%\CoverToEmbedding*
The rest of this section is devoted to the proof of \Cref{lem:CURWremoveAspectRatio}.
\begin{proof}[Proof of \Cref{lem:CURWremoveAspectRatio}]
Fix $s=\frac{8\cdot\rho\cdot\beta\cdot n}{\epsilon}$. For a parameter $\alpha>0$, let $G^{\alpha}=(V,E,w_{\alpha})$ be the graph with the
following weight function for every $e\in E$,
\[
w^{\alpha}(e)=\begin{cases}
	\alpha & w(e)\ge\alpha\\
	w(e) & \frac{\alpha}{s^{2}}<w(e)<\alpha\\
	0 & \frac{\alpha}{s^{2}}\ge w(e)
\end{cases}~.
\]
That is we reduce all edge weights to be at most $\alpha$, and we nullified (essentially contract)
all edges of weight at most $\frac{\alpha}{s^{2}}$. Note that as
${\cal F}$ is CURW, $G^{\alpha}\in{\cal F}$. We next analyze the
properties of $G^{\alpha}$.
\begin{claim}\label{claim:GalphaProp}
	The graph $G^{\alpha}$ has the following
	properties:
	\begin{enumerate}
		\item For every $x,y\in V$, $d_{G^{\alpha}}(x,y)\le\min\left\{ d_{G}(x,y),n\cdot\alpha\right\} $.
		\item For every $x,y\in X$ such that $d_{G}(x,y)\in[\frac{\alpha}{s},\alpha]$,
		$d_{G^{\alpha}}(x,y)\ge(1-\frac{n}{s})\cdot d_{G}(x,y)$.
		\item The induced shortest path metric has aspect ratio $n\cdot s^{2}$.
	\end{enumerate}
\end{claim}

\begin{proof}
	(1). The edge weights only decreased, so clearly $d_{G^{\alpha}}(x,y)\le d_{G}(x,y)$.
	In addition, the shortest path from $x$ to $y$ in $G^{\alpha}$
	consist of at most $n-1$ edges, all of weight at most $\alpha$,
	so it has weight at most $n\cdot\alpha$. (2). Consider a pair $x,y$
	such that $d_{G}(x,y)\in[\frac{\alpha}{s},\alpha)$. Let $P$ be the
	shortest path from $x$ to $y$ in $G^{\alpha}$. As $d_{G}(x,y)<\alpha$,
	we can assume that $P$ does not contain any edges of weight at least
	$\alpha$. It holds that
	\begin{equation}
		d_{G^{\alpha}}(x,y)=w^{\alpha}(P)>w(P)-n\cdot\frac{\alpha}{s^{2}}\ge d_{G}(x,y)-\frac{n}{s}\cdot d_{G}(x,y)~.\label{eq:RightScale}
	\end{equation}
	(3). Recall that our definition of aspect ratio ignores $0$ distances: $\Phi(G)=\frac{\max_{u,v\in V}d_G(u,v)}{\min_{u,v\in V~\rm{s.t.~}d_{\footnotesize G}(u,v)>0}d_G(u,v)}$ (see \Cref{sec:perlims}). 
	It follows that the minimal distance in $G^\alpha$ 
	is $\frac{\alpha}{s^{2}}$, and hence by point (1), the aspect ratio is bounded by $\frac{n\cdot\alpha}{\frac{\alpha}{s^{2}}}=n\cdot s^{2}$.
\end{proof}
Let $\alpha_{i}=\Phi\cdot s^{-i}$. For every $i\ge0$, construct an
embedding $f_{i}:X\rightarrow\ell_{\infty}$ w.r.t. $d_{G^{\alpha_{i}}}$
with expansion $\rho$ and contraction $\beta$ and dimension $\varphi(n,n\cdot s^{2},\rho,\beta)$.
We will then reuse coordinates such that all the embeddings $\left\{ f_{3i}\right\} _{i\ge0}$
all use the same coordinates. Similarly for $\left\{ f_{3i+1}\right\} _{i\ge0}$
and $\left\{ f_{3i+2}\right\} _{i\ge0}$. We will denote these embeddings
by $g_{0},g_{1},g_{2}$ respectively. This finishes the definition
of our embedding $f$. Clearly the overall number of coordinates is
bounded by $3\cdot\varphi(n,n\cdot s^{2},\rho,\beta)=3\cdot\varphi(n,\left(\frac{8\cdot\rho\cdot\beta}{\epsilon}\right)^{2}\cdot n^{3},\rho,\beta)$.
Next we bound expansion and contraction (see \Cref{fig:intervals} for illustration).

\begin{figure}[t]
	\centering
	\includegraphics[width=.8\textwidth]{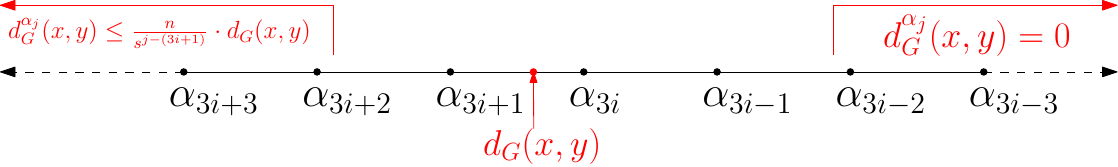}
	\caption{\footnotesize{Illustration of the contraction and expansion proof in \Cref{lem:CURWremoveAspectRatio}.
	}}
	\label{fig:intervals}
\end{figure}

Consider $x,y\in X$, and assume w.l.o.g. that there is $i$ such
that $\alpha_{3i+1}<d_{G}(x,y)\le\alpha_{3i}$ (the other cases is
where there is an $i$ such that either $\alpha_{3i+2}<d_{G}(x,y)\le\alpha_{3i+1}$
or $\alpha_{3i+3}<d_{G}(x,y)\le\alpha_{3i+2}$ are treated symmetrically).
For every $j\le3i-2$ (if any), it holds that $d_{G}(x,y)\le\alpha_{3i}=\alpha_{j}\cdot s^{j-3i}\le\alpha_{j}\cdot s^{-2}$.
Hence $d_{G^{\alpha_{j}}}(x,y)=0$, which implies $\|f_{j}(x)-f_{j}(y)\|_{\infty}=0$.
For every $j\ge3i+2$, 
\[
d_{G^{\alpha_{j}}}(x,y)\le n\cdot\alpha_{j}=n\cdot s^{(3i+1)-j}\cdot\alpha_{3i+1}\le n\cdot s^{(3i+1)-j}\cdot d_{G}(x,y)~.
\]
We conclude
\begin{align*}
	\|g_{0}(x)-g_{0}(y)\|_{\infty} & =\|\sum_{q\ge0}f_{3q}(x)-\sum_{q\ge0}f_{3q}(y)\|_{\infty}\\
	& \le\|f_{3i}(x)-f_{3i}(y)\|_{\infty}+\sum_{q\ge i+1}\|f_{3q}(x)-f_{3q}(y)\|_{\infty}\\
	& \le\rho\cdot d_{G}(x,y)+\rho\cdot\sum_{q\ge i+1}n\cdot s^{(3i+1)-3q}\cdot d_{G}(x,y)\\
	& =\rho\cdot d_{G}(x,y)+\rho\cdot\frac{n}{s^{2}}\cdot\frac{1}{1-s^{-3}}\cdot d_{G}(x,y)<\left(1+\frac{2n}{s^{2}}\right)\cdot\rho\cdot d_{G}(x,y)~.
\end{align*}
\begin{align*}
	\|g_{1}(x)-g_{1}(y)\|_{\infty} & \le\|f_{3i+1}(x)-f_{3i+1}(y)\|_{\infty}+\sum_{q\ge i+1}\|f_{3q+1}(x)-f_{3q+1}(y)\|_{\infty}\\
	& \le\rho\cdot d_{G}(x,y)+\rho\cdot\sum_{q\ge i+1}n\cdot s^{(3i+1)-(3q+1)}\cdot d_{G}(x,y)\\
	& =\left(1+\frac{n}{s^{3}}\cdot\frac{1}{1-s^{-3}}\right)\cdot\rho\cdot d_{G}(x,y)\le\left(1+\frac{2n}{s^{3}}\right)\cdot\rho\cdot d_{G}(x,y)
\end{align*}

\begin{align*}
	\|g_{2}(x)-g_{2}(y)\|_{\infty} & \le\|f_{3i-1}(x)-f_{3i-1}(y)\|_{\infty}+\sum_{q\ge i}\|f_{3q+2}(x)-f_{3q+2}(y)\|_{\infty}\\
	& \le\rho\cdot d_{G}(x,y)+\rho\cdot\sum_{q\ge i}n\cdot s^{(3i+1)-(3q+2)}\cdot d_{G}(x,y)\\
	& =\left(1+\frac{n}{s}\cdot\frac{1}{1-s^{-3}}\right)\cdot\rho\cdot d_{G}(x,y)\le\left(1+\frac{2n}{s}\right)\cdot\rho\cdot d_{G}(x,y)
\end{align*}
From the other hand, according to \Cref{claim:GalphaProp}, $d_{G^{\alpha_{3i}}}(x,y)\ge(1-\frac{n}{s})\cdot d_{G}(x,y)$.
It holds that
\begin{align*}
	\|g_{0}(x)-g_{0}(y)\|_{\infty} & \ge\|f_{3i}(x)-f_{3i}(y)\|_{\infty}-\sum_{q\ge i+1}\|f_{3q}(x)-f_{3q}(y)\|_{\infty}\\
	& \ge\frac{1}{\beta}\cdot d_{G^{\alpha_{3i}}}(x,y)-\rho\cdot\sum_{q\ge i+1}n\cdot s^{(3i+1)-3q}\cdot d_{G}(x,y)\\
	& =\frac{1}{\beta}\cdot(1-\frac{n}{s})\cdot d_{G}(x,y)-\rho\cdot\frac{n}{s^{2}}\cdot\frac{1}{1-s^{-3}}\cdot d_{G}(x,y)\\
	& >\left(1-\frac{n}{s}-\rho\cdot\beta\cdot\frac{2n}{s^{2}}\right)\cdot\frac{1}{\beta}\cdot d_{G}(x,y)~.
\end{align*}
As $\|f(x)-f(y)\|_{\infty}=\max_{j\in{0,1,2}}\|g_{j}(x)-g_{j}(y)\|_{\infty}$,
and according to our choice of $s$ we conclude that 
\begin{align*}
	\|f(x)-f(y)\|_{\infty} & \le\left(1+\frac{2n}{s}\right)\cdot\rho\cdot d_{G}(x,y)\le\left(1+\epsilon\right)\cdot\rho\cdot d_{G}(x,y)\\
	\|f(x)-f(y)\|_{\infty} & \ge\left(1-\frac{n}{s}-\rho\cdot\beta\cdot\frac{2n}{s^{2}}\right)\cdot\frac{1}{\beta}\cdot d_{G}(x,y)\ge\frac{1}{(1+\epsilon)\cdot\beta}\cdot d_{G}(x,y)
\end{align*}

\end{proof}

\subsection{Improved distortion for minor free graphs}\label{subsec:ImproveEmbeddingMinorFree}
This subsection is devoted to proving \Cref{thm:EmbeddingMinor3}, which improves the distortion of the $\ell_\infty$ embedding of $K_r$-minor free graphs to $3+\eps$. The following is our key lemma, which is parallel to \Cref{lem:fromSparseCoverToEmbedding}.
\begin{lemma}\label{lem:fromSparseCoverToEmbeddingMinorFree}
	Fix $\eps\in(0,\frac12)$ and $\Phi\ge1$, and consider a weighted $n$ vertex $K_r$-minor free graph $G=(V,E,w)$ such that all the pairwise distances in $X$ are between $1$ and $\Phi$. Then $G$ can be embedded into $\ell_\infty^{O(\frac1\eps)^{r+1}\cdot\log\frac{1}{\eps}\cdot\log(n\Phi)}$ with distortion $3+\eps$.
%	
%	Consider an $n$-point metric space $(X,d_X)$ that admits a $(\beta,\tau)$-sparse partition cover scheme, and such that all the pairwise distances in $X$ are between $1$ and $\Phi$. Then for every  $\eps\in(0,1)$, there is an (efficiently computable) embedding $f:X\rightarrow\ell_\infty^k$ for $k=O\left(\frac{\tau}{\eps}\cdot\log\frac{\beta}{\eps}\cdot\log(n\Phi)\right)$ and distortion $(1+\eps)\cdot2\beta$.
\end{lemma}

%, restated below for convenience. 
%\EmbeddingMinorFree*
\begin{proof}
	The embedding is essentially the same as in the construction of \Cref{lem:fromSparseCoverToEmbedding}, and the main changes are in the analysis. 
	Roughly speaking, first we subdivide all the edges of our graph, this will reduce the expansion from $2$ to $1+\eps$. Then we will dive into the specific details of the sparse cover, and show that the contraction is only $3+O(\eps)$.
	
	Consider our $K_r$ minor free graph $G=(V,E,w)$.
	Let  $\tilde{G}=(\tilde{V},\tilde{E},\tilde{w})$ be a graph, where we replace every edge 
	 $e=(v,u)\in E$ with a path consisting of $\frac1\eps$ edges of length $\eps\cdot w(e)$. Note that all the pairwise distances between original vertices remained exactly the same, and that $\tilde{G}$ is still $K_r$-minor free, thus it also admits the sparse covers from \Cref{thm:MinorFreeCover}. We will embed only the original vertices $V$ into $\ell_\infty$, however the partitions will be created w.r.t. $\tilde{G}$.
	This small change will improve our upper bound in \Cref{lem:singleHierarchy} from expansion  $2$ to expansion $1+\eps$. Indeed, consider a pair $x,y\in V$, and recall the proof of \Cref{lem:singleHierarchy}.
	The proof is by induction on $|V|$. 
	The case where $x,y$ belong to the same cluster in $\cP_k$ is treated in the same way (as \cref{eq:Lipshitz1} guarantees $\left|(f(x))_{q}-(f(y))_{q}\right|\le d_{G}(x,y)$ ).
	In the second case, $x$ and $y$ correspond so different clusters, where according to \cref{eq:Lipshitz2}, $\left|(f(x))_{q}-(f(y))_{q}\right|\le d_{\tilde{G}}(x,\tilde{V}\setminus Y_{j})+d_{\tilde{G}}(y,\tilde{V}\setminus Y_{j'})$ where $x\in Y_{j}$, $y\in Y_{j'}$, and $Y_{j}\cap Y_{j'}=\emptyset$ are disjoint. 
	In $\tilde{G}$, the shortest path $P$ from $x$ to $y$ in $\tilde{G}$ must go though $x'\notin Y_j$, and $y'\notin Y_{j'}$ such that $x',y'$ are consecutive vertices of a subdivided edge $e$ (see two illustrations bellow). In particular,
	 $d_{\tilde{G}}(x',y')=\eps\cdot w(e)\le\eps\cdot w(P)$. It follows that 
	 
	\begin{align*}
		\left|(f(x))_{q}-(f(y))_{q}\right| & \le d_{\tilde{G}}(x,\tilde{V}\setminus Y_{j})+d_{\tilde{G}}(y,\tilde{V}\setminus Y_{j'})\\
		& \le d_{\tilde{G}}(x,x')+d_{\tilde{G}}(y,y')~\le~d_{\tilde{G}}(x,y)+d_{\tilde{G}}(x',y')~<~(1+\eps)\cdot d_{\tilde{G}}(x,y)~.
	\end{align*}
	\begin{center}
		\includegraphics[width=.9\textwidth]{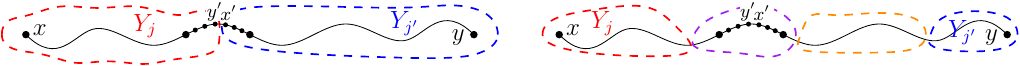}	 
	\end{center}
	It follows that the embedding produced by \Cref{thm:SPCStoEmbeddingNoAspect} has expansion $1+\eps$.

	Next we improve the contraction guarantee. The basis of our construction is the $\left(4\cdot(1+\eps),O(\frac{1}{\eps})^{r}\right)$-\SPCS of \Cref{thm:MinorFreeCover}.
	We begin by observing that the padding guarantee in \Cref{thm:ReductionBufferedToCovers} is somewhat stronger that what is required by the definition of $\SPCS$. 
	Indeed, we follow the construction (and notation) of \Cref{thm:ReductionBufferedToCovers}. Note that here  $\gamma=\frac\Delta r$, and for the construction of our strong $\left(4\cdot(1+\eps),O(\frac{1}{\eps})^{r}\right)$-\SPCS  we used $q=\frac{2r}{\eps}$. 
	In \Cref{lem:CoverByCores,lem:SupernodePartitions} it was shown that the cluster satisfying a point $x$ consist of a ball $B_{G[\hat{\eta}]}(v,2\Delta+q\cdot\gamma)=B_{G[\hat{\eta}]}(v,(1+\eps)\cdot\frac{2\Delta}{\eps})$
	where $B_{G}(x,\frac{q\cdot\gamma}{2})=B_{G}(x,\frac{\Delta}{\eps})\subseteq\hat{\eta}$
	and $d_{G[\hat{\eta}]}(v,x)\le2\Delta+\frac{q\cdot\gamma}{2}=(1+2\eps)\cdot\frac{\Delta}{\eps}$.
	Note that the diameter of the corresponding partition is $(1+\eps)\cdot\frac{4\Delta}{\eps}$, and that $\partial_{\cP}(x)\ge\frac{\Delta}{\eps}$. 
	In other words, fix $\Psi=\frac{4\Delta}{\eps}$, then there is a cluster $C_x\in\cP$ such that $B_G(x,\Psi)\subseteq C\subseteq B_G(v,(1+\eps)\cdot 2\Psi)$, and $d_G(x,v)\le(1+2\eps)\cdot\Psi$.
	Next, in \Cref{lem:FromCoverToLaminar}, the partition $\cP$ was slightly changed to into $\tilde{\cP}$ in order to become a part of an hierarchy. In particular, it holds that there is a cluster $\tilde{C}\in \tilde{\cP}$ such that $B_G(x,\frac{\Psi}{1+\eps})\subseteq \tilde{C}\subseteq B_G(v,(1+\eps)^2\cdot 2\Psi)$. 
	Furthermore, for a vertex $y\in V$  such that $d_G(x,y)>(3+8\eps)\cdot\Psi$ it holds that
	\[
	d_{G}(v,y)\ge d_{G}(x,y)-d_{G}(v,x)>(3+8\eps)\cdot\Psi-(1+2\eps)\cdot\Psi=(1+3\eps)\cdot2\Psi>(1+\eps)^{2}\cdot2\Psi~,
	\]
	and thus $y$ does not belong to $\tilde{C}$.
	
	Fix a pair $x,y\in V$. In \Cref{lem:fromSparseCoverToEmbedding} we constructed partition for every possible scale. Consider the scale $(i,q)$ such that $(1+\epsilon)^{q}(\frac{4\beta}{\epsilon})^{i}<\frac{4}{3+8\cdot \eps}\cdot d_{G}(x,y)\le(1+\epsilon)^{q+1}(\frac{4\beta}{\epsilon})^{i}$. 
	Fix $\Psi=\frac14\cdot(1+\epsilon)^{q}(\frac{4\beta}{\epsilon})^{i}$.
	Following the discussion above, 
	for this scale, there was some partition $\cP_i^{j,q}$ with diameter at most $(1+\eps)^2\cdot4\Psi$ such that the ball $B_G(x,\frac{\Psi}{1+\eps})$ was contained in a single cluster, and every vertex at distance greater than $(3+8\eps)\cdot\Psi$ from $x$ belonged to a different cluster. In particular, as $d_{G}(x,y)>\frac{3+8\cdot\eps}{4}\cdot(1+\epsilon)^{q}(\frac{4\beta}{\epsilon})^{i}=(3+8\cdot\eps)\cdot\Psi
	$, $y$ belongs to a different cluster.
	Using \Cref{lem:singleHierarchy} we conclude
	\begin{align*}
	 	\left\Vert f(x)-f(y)\right\Vert _{\infty} & \ge\left\Vert f_{j,q}(x)-f_{j,q}(y)\right\Vert _{\infty}\ge\partial_{\cP_{i}^{j,q}}(x)+\partial_{\cP_{i}^{j,q}}(y)\\
	 	& \ge\frac{\Psi}{1+\eps}=\frac{(3+8\eps)\cdot\Psi}{(1+\eps)\cdot(3+8\eps)}>\frac{d_{G}(x,y)}{3\cdot(1+O(\eps))}~.
	\end{align*}
	Thus we indeed obtained distortion $3+O(\eps)$ as promised. Similarly to \Cref{lem:fromSparseCoverToEmbedding}, the overall dimension is bounded by $
	O\left(\frac{\tau}{\eps}\cdot\log\frac{\beta}{\eps}\cdot\log(n\Phi)\right)=O(\frac{1}{\eps})^{r+1}\cdot\log\frac{1}{\eps}\cdot\log(n\Phi)$.
\end{proof}
	Finally, we apply \Cref{lem:CURWremoveAspectRatio} on \Cref{lem:fromSparseCoverToEmbeddingMinorFree} to get \Cref{thm:EmbeddingMinor3}, restated bellow for convenience. 
	\EmbeddingMinorFree*

\section{Oblivious Buy-at-Bulk}\label{subsec:BuyAtBulk}
In the \emph{buy-at-bulk} problem we are given a weighted graph $G=(V,E,w)$, the goal is to satisfy a set of demands $A\subseteq{V\choose 2}$, by routing these demands over the graph while minimizing the cost. 
In more details, $\delta_i=(s_i,t_i)$ is a unit of demand that induces an unsplittable unit of flow from source node $s_i\in V$ to a destination $t_i\in V$. Given a set of demands $A=\{\delta_1,\delta_2,\dots,\delta_k\}$, a valid solution is a set of paths $\cP=\{P_1,\dots,P_k\}$, where $P_i$ is a path from $s_i$ to $t_i$. The paths can overlap. For an edge $e\in E$, denote by $\varphi_e$ the number of paths in $\cP$ that use $e$.
A function $f:\N\rightarrow\R_{\ge0}$ is called \emph{canonical} fusion function if it is (1) concave, (2) non-decreasing, (3) $f(0)=0$, and (4) sub-additive (that is $\forall x,y\in \N$, $f(x+y)\le f(x)+f(y)$).
In the buy-at-bulk problem we are given a canonical fusion function $f$. The cost of a solution $\cP$ is $\cost(\cP)=\sum_{e\in E}f(\varphi_e)\cdot w(e)$. The goal is to find a valid solution of minimum cost. 
The canonical fusion function provides the following intuition: there is a discount as you use more and more of the same edge. So it is generally beneficial for the paths in $\cP$ to overlap. 
Note that the buy-at-bulk problem is NP-hard, as the Steiner tree problem is a special case (where the canonical fusion function is $f(0)=0$ and $f(i)=1$ for $i\ge1$).

In the \emph{oblivious buy-at-bulk} problem we have to choose a collection of paths $\cP$ without knowing the demands. That is, for every possible demand $\delta_i\in {V\choose 2}$, we have to add to $\cP$ choose a path $P(\delta_i)$ between it's endpoints. Then given a set of demands $A\subseteq{V\choose 2}$, 
$\cost(\cP,A)=\sum_{e\in E}f(\varphi_e(\cP,A))\cdot w(e)$, where $\varphi_e(\cP,A)=\left|\{P(\delta_i)\mid e\in P(\delta_i)~\&~\delta_i\in A\}\right|$ is the number of paths associated with the demands in $A$ that use $e$.
The approximation ratio of $\cP$, is the ratio between the induced cost, to the optimal cost for the worst possible set of demands $A$:
\[{\rm Approximatio~ ratio}(\cP)=\max_{A\subseteq{V\choose 2}}\frac{\cost(\cP,A)}{\opt(A)}~.\]
Gupta \etal \cite{GHR06} provided an algorithm achieving approximation ratio $O(\log^2n)$. Interestingly, their algorithm is oblivious not only to the demand pairs $A$, but also to the canonical fusion function $f$.
A lower bound of $\Omega(\log n)$ is known already for the case where the graph $G$ is planar \cite{IW91}.
Later, Srinivasagopalan \etal \cite{SBI11} improved the approximation ratio for the case of planar graphs to $O(\log n)$. Their solution is also oblivious to both demands and function. 
Srinivasagopalan \etal \cite{SBI11} left it as an explicit open problem to 
``obtain efficient solutions to other related network topologies, such as minor-free graphs.'' More than a decade later, and compared with general graphs, nothing better for $K_r$-minor free graphs is known.

Our contribution here is to answer the question from \cite{SBI11} affirmatively. Specifically, that every  $K_r$-minor free graph admits a solution with approximation ratio $\poly(r)\cdot\log n$. This is tight up to the dependence on $r$.

Srinivasagopalan \etal \cite{SBI11} implicitly proved that if a graph admits a ``colorable'' strong sparse cover scheme, than it also admits an efficiently commutable solution with small approximation ratio for the oblivious buy-at-bulk problem.
\begin{definition}\label{def:colorableCover}
	Consider a strong $(\beta,s,\Delta)$-sparse cover $\cC$ of $G=(V,E,w)$.
	We say that cluster $C\in \cC$ $\beta$-satisfies a vertex $x\in V$ if the ball $B_G(x,\frac\Delta\beta)\subseteq C$ is contained in $C$.
	Two clusters $C_1,C_2\in\cC$ are neighbors if there is a pair of vertices $x\in C_1$ and $y\in C_2$ at distance at most $d_G(x,y)\le \frac\Delta\beta$ such that $x,y$ are $\beta$-satisfied by $X_1,X_2$ respectively.
	A $k$-coloring is a function $\chi:\cC\rightarrow [k]$ such that no two neighboring clusters are colored using the same color. 
\end{definition}
Srinivasagopalan \etal \cite{SBI11}  showed that the strong sparse cover of \cite{BLT14} for planar graphs is $18$-colorable. Note that given a strong $(\beta,s,\Delta)$-sparse partition cover, it is clearly $s$-colorable. Indeed, we can give all the clusters in each partition the same color. As the clusters in each partitions are all disjoint, this is a valid coloring. We conclude that our sparse covers from \Cref{thm:MinorFreeCover} are $O(r^2)$-colorable.

Srinivasagopalan \etal \cite{SBI11} showed that if a graph $G$ admits a strong $(\beta,s)$-sparse cover scheme, where each cover is $k$-colorable, than one can efficiently compute a solution for the oblivious buy-at-bulk problem with approximation ratio $O(s\cdot\beta^2\cdot k\cdot\log n)$. 
Interestingly, no colorable strong sparse covers for $K_r$-minor free graphs were known before. 
Indeed, the previous covers were either not strong \cite{KPR93,FT03}, or not colorable \cite{AGMW10,AGGNT19,Fil19padded}.      
Hence the question from  \cite{SBI11} remained open until now.
Using our \Cref{thm:MinorFreeCover} we conclude:
%\begin{corollary}\label{cor:MinorBuyAtBulk}
%	For every $n$-vertex weighted $K_r$-minor free graph $G=(V,E,w)$ admits an efficiently commutable solution to the oblivious buy-at-bulk problem with approximation ratio $O(r^6\cdot\log n)$. Furthermore, the solution is also oblivious to the concave faction $f$.
%\end{corollary}

\BuyAtBulk*
%Srinivasagopalan \etal \cite{SBI11}, implicitly showed that if a graph $G$ admits a strong $(\beta,s)$-sparse cover scheme, where each cover is $k$-colorable (see \Cref{def:colorableCover}), than one can efficiently compute a solution for the oblivious buy-at-bulk problem with approximation ratio $O(s\cdot\beta^2\cdot k\cdot\log n)$. 
%We observe that our sparse cover scheme from \Cref{thm:MinorFreeCover} is $O(r^2)$-colorable, and conclude that $K_r$ minor free graphs admit a solution with approximation ratio $O(r^6\cdot\log n)$ (see \Cref{cor:MinorBuyAtBulk}), answering an open problem from \cite{SBI11}.
%Note that nothing better for minor free graphs compared with general graphs  was previously known. This is a previous covers were either not strong \cite{KPR93,FT03}, or not colorable \cite{AGMW10,AGGNT19,Fil19padded}.
\section{Further applications}
\subsection{Sparse Partitions}\label{subsec:SparsePartition}
Given a weighted graph $G=(V,E,w)$, a weak/strong $(\alpha, \tau,\Delta)$-sparse partition is a partition $\cC$ of $V$ such that:
\begin{itemize}
	\item \textbf{Low Diameter:} $\forall C\in\cC$, the $C$ has a  weak/strong  diameter at most $\Delta$;
	\item \textbf{Ball Preservation:} $\forall v \in V$, the ball $B_G(v, \frac{\Delta}{\alpha})$ intersects at most $\tau$ clusters from $\cC$.
\end{itemize}
We say that the graph  $G=(V,E,w)$ admits a  weak/strong  $(\alpha, \tau)$-sparse partition scheme if for every $\Delta>0$, $G$ admits a  weak/strong $(\alpha, \tau,\Delta)$-sparse partition.

Sparse partition have been used to construct universal TSP and universal Steiner tree  \cite{JLNRS05,BDRRS12,Fil20,BCFHHR24} (see \Cref{subsec:UST}). Recently they been also used to obtain an $O(\frac{d}{\log d})$ approximation for the facility location problem in $\R^d$ in the geometric streaming model \cite{CJFKVY23}.

Jia \etal \cite{JLNRS05} implicitly proved (see \cite{Fil20} for an explicit proof) that if a space admits a weak $(\beta,s)$-sparse cover scheme, then it admitsa weak $(\beta,s)$-sparse partition scheme. Using \Cref{thm:MinorFreeCover} we conclude:
%\begin{restatable}[]{corollary}{SparseCover}
%	\label{cor:sparsePartition}
%	Every  $K_r$-minor free graph $G$ admits a weak $\left(O(r),O(r^2)\right)$-sparse partition scheme. In addition, for $\eps\in(0,\frac12)$, $G$ admits a weak $\left(4+\eps,O(\frac1\eps)^r\right)$-sparse partition.
%\end{restatable}
\begin{corollary}\label{cor:sparsePartition}	
	Every  $K_r$-minor free graph $G$ admits a weak $\left(O(r),O(r^2)\right)$-sparse partition scheme. In addition, for $\eps\in(0,\frac12)$, $G$ admits a weak $\left(4+\eps,O(\frac1\eps)^r\right)$-sparse partition.
\end{corollary}
The previous state of the art follows directly from the previously best know weak sparse cover scheme, that is weak $\left(O(r^2),O(2^r)\right)$-sparse partition scheme (see \cite{Fil20}, and also \cite{KPR93,FT03}). 

\subsection{Universal \TSP and Universal Steiner Tree}\label{subsec:UST}
%\atodo{Mention the new treewidth graph? We kind of improve over it here}
Consider a postman providing post service for a set $X$ of clients with $n$ different locations (with distance measure $d_X$). Each morning the postman receives a subset $S\subseteq X$ of the required deliveries for the day. In order to minimize the total tour length, one solution may be to compute each morning an (approximation of an) optimal \TSP tour for the set $S$. An alternative solution will be to compute a \emph{Universal \TSP} (\UTSP) tour. This is a universal tour $R$ containing all the points $X$. Given a subset $S$, $R(S)$ is the tour visiting all the points in $S$ w.r.t. the order induced by $R$.
Given a tour $T$ denote its length by $|T|$. The \emph{stretch} of $R$ is the maximum ratio among all subsets $S\subseteq X$ between the length of $R(S)$ and the length of the optimal \TSP tour on $S$, $\max_{S\subseteq X}\frac{|R(S)|}{|\mbox{Opt}(S)|}$.

A closely related problem to \UTSP is the \emph{Universal Steiner tree} (\UST).
%Here the goal is to design
Consider the problem of designing a network that allows a server to broadcast a message to a single set of clients. If sending a message over a link incurs some cost then designing the best broadcast network is classically modeled as the Steiner tree problem \cite{Hwang76}. However, if the server have to solve this problem repeatedly, with different client sets, it is desirable to construct a single network that will optimize the cost of the broadcast for every possible subset of clients.

Given a metric space $(X,d_X)$ and root  $r \in X$, a $\rho$-approximate universal Steiner tree (\UST) is a weighted tree $T$ over $X$ such that for every $S \subseteq X$ containing $r$, we have
\begin{align*}
	w(T\{S\}) \leq \rho \cdot \OPT_S
\end{align*}
where $T\{S\} \subseteq T$ is the minimal subtree of $T$ connecting $S$, and $\OPT_S$ is the minimum weight Steiner tree connecting $S$ in $X$. If the tree $T$ is not required to be a subgraph of $G$, we will call the problem metric \UST.

Jia \etal \cite{JLNRS05} showed that for every $n$-point metric space that admits weak $(\sigma,\tau)$-sparse partition scheme, there is a polynomial time algorithm that given a root $\rt\in V$ computes 
a \UTSP with stretch  $O(\tau\sigma^2\log_\tau n)$, and 
a metric \UST with stretch $O(\tau\sigma^2\log_\tau n)$.
%In addition, Jia \etal \cite{JLNRS05} also showed that such a metric admit a .
Using our \Cref{cor:sparsePartition}, we conclude:
\begin{corollary}\label{cor:UST}
	Consider an $n$-point $K_r$-minor free weighted graph $G=(V,E,w)$. Then $G$ admits a solution to both \UTSP and metric \UST with stretch $O(\frac{r^4}{\log r})\cdot\log n$. 
\end{corollary}
The previous state of the art had stretch $O(2^r\cdot r)\cdot\log n$, thus we obtain an exponential improvement in the dependence on $r$.
Interestingly, there is an $\Omega(\frac{\log n}{\log\log n})$ lower bound for the \UST problem on the $n\times n$ grid, which is $K_5$-minor free, thus essentially, in the context of $K_r$-minor free graphs, the dependence on $r$ is the only parameter left to optimize.

It is important to note that the \UST returned by \cite{JLNRS05} (and thus by \Cref{cor:UST}) is not a subgraph of $G$. 
The (non-metric) \UST problem was also studied and there is a similar reduction \cite{BDRRS12}. Formally, a strong $(\tau,\sigma,\rho)$-sparse partition hierarchy is a set of laminar partitions $\cP_1,\cP_2,\dots$, such that $\cP_i$ refines $\cP_{i+1}$, and each  $\cP_i$ is a strong $(\tau,\sigma,\rho^i)$-sparse partition. 
Busch \etal \cite{BDRRS12} showed that if a graph admits strong $(\tau,\sigma,\rho)$-sparse partition than there is a subgraph solution to the \UST problem with stretch $O(\tau^2\cdot\sigma^2\cdot\rho\cdot\log n)$. This reduction was recently used by Busch \etal \cite{BCFHHR24} to obtain a solution for the \UST problem on general graphs with stretch $O(\log^7n)$. \cite{BCFHHR24} also showed solution with stretch $O(\log n)$ for graph with constant doubling dimension or pathwidth. However, nothing better than general graphs is known for $K_r$-minor free graphs, or even planar graphs. One reason being that even trees do no admit better than $(\tilde{\Omega}(\log n),\tilde{\Omega}(\log n))$-strong sparse partition scheme \cite{Fil20scattering}. Improving the stretch factor for $K_r$-minor free graphs is a fascinating open problem.

%The best known previous result for both problems had stretch $\exp(\tw)\cdot\log n$ \cite{Fil20}.
%For further reading on the \UTSP and \TSP problems, see 
%\cite{platzman1989spacefilling,bertsimas1989worst,JLNRS05,GHR06,HKL06,schalekamp2008algorithms,gorodezky2010improved,bhalgat2011optimal,BDRRS12,BLT14,Fil20}.\atodo{Cite here new UST paper.}

\begin{remark}
	Universal Steiner tree and Oblivious buy-at-bulk problems sound very similar. Indeed, one might be tempted to think that the Universal Steiner tree problem is a special case where there is only a single source, and the canonical function gets only the values $\{0,1\}$. However, in the universal Steiner tree we have additional demand- the solution has to be a tree! Indeed, at present, there are significant gaps between the best solutions for this problems on both general, and $K_r$-minor free graphs.
\end{remark}

\subsection{Name-Independent Routing}\label{subsec:routing}
A \emph{compact routing scheme} in a network is a mechanism that allows packets to be delivered from any node to any other node. The network is represented as a undirected graph, and each node can forward incoming data by using local information stored at the node, called a \emph{routing table}, and the (short) packet's \emph{header}. The {\em stretch} of a routing scheme is the worst-case ratio between the length of a path on which a packet is routed to the shortest possible path. When designing a compact routing scheme our goal is to optimize the trade-off between the stretch and the size of the routing table and header.
Here we focus on the name-independent model where each node is assigned an arbitrary unique network identifier, which cannot be chosen by the routing scheme designer. 
In addition, each edge will have an arbitrary port (that also cannot be changed).
Note that there is a different regime of labeled routing, where the designer can choose the node labels and edge ports, see e.g. \cite{TZ01b,Thorup04,AG06,Chechik13,ACEFN20,Fil21,FL21}. The labeled routing regime admits much better compact routing schemes. Indeed for general graphs, for every $k\in\N$ one can construct labeled routing scheme with stretch $3.68k$, routing table size $O(k\cdot n^{\frac1k})$, and label size $O(k\cdot\log n)$ (see also \cite{ACEFN20,FL21} for different trade-off's). In $K_r$-minor free graphs one can even construct labeled routing scheme with stretch $1+\eps$ and and label and table size of $\eps^{-1}\cdot f(r)\cdot\polylog(n)$  $^{\ref{foot:RS}}$ \cite{AG06}.

The name-independent regime is considered more practical, as it does not assumes that the sender knows the artificially chosen label of the destination. Further, it can cope much better with changes in the network. There are also applications that pose constrains on nodes addresses that are not easily satisfied by the artificial smartly chosen labels (such as distributed hash tables).
However, the  name-independent regime is much more challenging.
Indeed, any name-independent routing scheme on unweighted stars (trees of depth $1$) requires table size of $\Omega(n\cdot\log n)$ bits to get stretch bellow $3$ \cite{AGMNT08}. The situation in weighted stars is even more dire, as every name-independent routing scheme with stretch bellow $2k+1$ (for $k\ge 1)$ requires routing tables of at least $\Omega((n\cdot\log n)^{\frac1k})$ bits \cite{AGM06LB}. We will thus focus on unweighted graphs.

Abraham \etal \cite{AGMW10} used their strong sparse covers, as well as name-independent routing scheme for trees \cite{AGM04}, to design a name-independent routing scheme for $K_r$-minor free graphs. We can use our new strong sparse cover to improve the various parameters in this routing scheme. In the following we elaborate on that.

A hereditary graph family is a family ${\cal F}$ such that for every
$G\in{\cal F}$, every subgraph $H$ of $G$ also belongs to $H\in{\cal F}$.
A graph posses an $\alpha$-orientation if it is possible to direct all
edges such that every vertex has out-degree at most $\alpha$. A family
${\cal F}$ has an $\alpha$-orientation if every $G\in{\cal F}$
has $\alpha$-orientation. 
\cite{AGMW10} implicitly proved the following meta theorem: consider a hereditary
graph family ${\cal F}$ that has an $\alpha$-orientation, and let
$G\in{\cal F}$ be an unweighted graph with diameter $D$ that
admits a strong $(\tau,\beta)$-sparse cover scheme. Then there is
a polynomial time constructible name-independent routing seheme, in
the fixed port model with stretch $O(\beta)$, and using $O(\log n+\log\tau)$-bit
headers, in which every node requires tables of $O(\frac{\log^{3}n}{\log\log n}+\alpha\cdot\log n)\cdot\tau\cdot\log D$
bits.

It is well known that the hereditary graph family of $K_{r}$ minor
free graphs has an $O(r\sqrt{\log r})$-orientation (see e.g. \cite{AGMW10}).
Combining it with our \Cref{thm:MinorFreeCover} we conclude
\begin{corollary}\label{cor:Labeling}
	For every $n$ vertex unweighted $K_{r}$-minor free graph $G$ with
	diameter $D$, there is a polynomial time constructible name-independent
	routing scheme, in the fixed port model, with:
	\begin{itemize}
		\item Stretch $O(r)$, $O(\log n)$ bit headers, and $O(\frac{\log^{3}n}{\log\log n}+r\sqrt{\log r}\cdot\log n)\cdot r^{2}\cdot\log D$
		bit table.
		\item Stretch $O(1)$, $O(r+\log n)$ bit headers, and $O(1)^{r}\cdot\frac{\log^{3}n}{\log\log n}\cdot\log D$
		bit table.
	\end{itemize}
\end{corollary}
See \Cref{tab:routing} for  comparison of \Cref{cor:Labeling} with previous results.
\begin{table}[]
	\begin{tabular}{l|l|l|l|l|}
		\cline{2-5}
		&\textbf{Stretch} & \textbf{Header}     & \textbf{Table}                                                                      & \textbf{Ref} \\ \hline
		\multicolumn{1}{|l|}{(1)}&$O(r^{2})$       & $O(\log n+r\log r)$ & $O(1)^{r}\cdot r!\cdot\frac{\log^{3}n}{\log\log n}\cdot\log D$                      & \cite{AGMW10}       \\ \hline
		\multicolumn{1}{|l|}{(2)}&$O(1)$           &                  & $O(f(r)\cdot\log^3n)$                                                               & \cite{BLT14}        \\ \hline
		\multicolumn{1}{|l|}{(3)}&$O(r)$           & $O(\log n)$         & $O(\frac{\log^{3}n}{\log\log n}+r\sqrt{\log r}\cdot\log n)\cdot   r^{2}\cdot\log D$ & \Cref{cor:Labeling} \\ \hline
		\multicolumn{1}{|l|}{(4)}&$O(1)$           & $O(\log n+r)$       & $O(1)^{r}\cdot\frac{\log^{3}n}{\log\log n}\cdot\log D$                              & \Cref{cor:Labeling} \\ \hline
	\end{tabular}
	\caption{\small{Summery of new and previous work on name-independent
			routing schemes in $K_r$-minor free graphs in the fixed parameter model. The input graph are unweighted, and $D$ denotes it's diameter.  $f(r)$ is an extremely fast growing function $^{\ref{foot:RS}}$. In \cite{BLT14}, the header size is not explicitly stated, and the bound on the table is only in expectation. Comparing our result (3) with \cite{AGMW10} we quadratically improved the stretch, and exponentially improved the dependence on $r$ in the routing table size. Comparing \cite{AGMW10} with (4) we improved the stretch from $O(r^2)$ to constant, and somewhat improved the dependence on $r$ in both table size and header. Comparing (4) with \cite{BLT14}, we very significantly improved the table size, and improved the size bound to hold in the worst case.
	}}
	\label{tab:routing}
\end{table}

\subsection{Path-Reporting Distance Oracle}\label{subsec:PathReporting}
Given a weighted graph $G=(V,E,w)$, a path-reporting distance oracle (\PRDO) is a succinct data structure that given a query $\{u,v\}\in{V\choose 2}$ quickly returns a $u-v$ path which is an approximate $u-v$ shortest path.
Formally we say that a distance oracle has stretch $k$ if on query $\{u,v\}$, it returns a $u-v$ path $P$ of weight at most $k\cdot d_G(u,v)$. We say that the query time is $t$, if the time it takes the distance oracle to return a path $P$ is at most $O(|P|)+t$. 
In the study of \PRDO's \cite{EP16,ENW16,ES23}, the goal is to optimize the trade-off between stretch, query time, and space.

Elkin, Neiman and Wulff-Nilsen \cite{ENW16} constructed a \PRDO using strong sparse covers with very small query time. 
Given a weighted graph $G=(V,E,w)$ with aspect ratio $\Phi$ that admits a strong $(\beta,s)$-sparse cover scheme, \cite{ENW16} implicitly constructed a \PRDO with stretch $O(\beta)$, space $O(s\cdot\log\Phi)$ (in machine words), and query time $O(\log\log\Phi)$. In fact, by decreasing the gap between the different scales in their construction to $1+\eps$, their proof will lead to a \PRDO with stretch $(1+\eps)\cdot\beta$, space $O(\frac{s}{\eps}\cdot\log\Phi)$, and query time $O(\log\log\Phi+\log\frac1\eps)$.

Elkin \etal \cite{ENW16} used the strong $(O(r^2),O(\log n))$-sparse cover scheme that follows from the padded decomposition of \cite{AGGNT19} to construct their \PRDO for $K_r$-minor free graphs with stretch $O(r^2)$, size $O(n\cdot\log n\cdot\log\Phi)$ and query time $O(\log\log\Phi)$.
Later, Filtser \cite{Fil19padded} used his improved strong padded decomposition to obtain a strong $(O(r),O(\log n))$-sparse cover scheme, and plugged it into the \PRDO construction of \cite{ENW16}, improving the stretch parameter from $O(r^2)$ to $O(r)$.
Using our strong $(O(r),O(r^2))$-sparse cover scheme (\Cref{thm:MinorFreeCover}), we can get a further improvement, reducing the space to  $O(n\cdot r^2\cdot\log\Phi)$.
Furthermore, using our strong $(4+\eps,O(\frac1\eps)^r)$-sparse cover scheme (\Cref{thm:MinorFreeCover}), and the observations made in the proof of the contraction argument in \Cref{lem:fromSparseCoverToEmbeddingMinorFree}, we can obtain a \PRDO with stretch $3+\eps$, size $n\cdot O(\frac1\eps)^{r+1}\cdot\log\Phi$ and query time $O(\log\log\Phi+\log\frac1\eps)$.

%For the case of $K_r$-minor free graph with aspect ratio $\Phi$, Elkin, Neiman and Wulff-Nilsen \cite{ENW16} constructed a path-reporting distance oracle with stretch $O(r^2)$, size $O(n\cdot\log n\cdot\log_r\Phi)$ and query time $O(\log\log\Phi)$. Their construction is based on the strong $(O(r^2),O(\log n))$-sparse cover scheme that follows from the padded decomposition of \cite{AGGNT19}. 
%Later, Filtser \cite{Fil19padded} used his improved strong padded decomposition to obtain a strong $(O(r),O(\log n))$-sparse cover scheme, and plugged it into the path-reporting distance oracle of \cite{ENW16}, improving the stretch parameter from $O(r^2)$ to $O(r)$.
%Using our strong $(O(r),O(r^2))$-sparse cover scheme (\Cref{thm:MinorFreeCover}), we can get a further improvement, reducing the space to  $O(n\cdot r^2\cdot\log_r\Phi)$.
\begin{corollary}\label{cor:PathReporting}
	Every $n$-vertex $K_r$-minor free weighted graph $G$ with aspect ratio $\Phi$, admits a path-reporting distance oracle with:
	\begin{itemize}
		\item Stretch $O(r)$, size $O(n\cdot r^2\cdot\log\Phi)$ and query time $O(\log\log\Phi)$. 
		\item Stretch $3+\eps$, size $n\cdot O(\frac1\eps)^{r+1}\cdot\log\Phi$ and query time $O(\log\log\Phi+\log\frac1\eps)$. 
	\end{itemize} 
\end{corollary}

%Now we can plug in out new $(O(r),O(r^2))$ sparse cover to improve the space further from $O(n\cdot\log n\cdot\log_r\Phi)$ to $O(n\cdot r^2\cdot\log_r\Phi)$.
%
%The basic theorem is that given a strong $(\beta,\tau)$-sparse cover scheme, one can construct a path reporting distance oracle with stretch $\beta$, size $O(n\cdot s\cdot\log_\beta\Phi)$ and query time $O(\log\log\Phi)$. 
%Then based on \cite{AGGNT19} they constructed an $(O(r^2),O(\log n))$ sparse cover to obtain 
%path reporting distance oracle with stretch $O(r^2)$, size $O(n\cdot\log n\cdot\log_r\Phi)$ and query time $O(\log\log\Phi)$. 
%Later, Filtser \cite{Fil19padded} used his strong padded decomposition to obtain an $(O(r),O(\log n))$ sparse cover to improve the stretch of  \cite{ENW16} from $O(r^2)$ to $O(r)$.
%Now we can plug in out new $(O(r),O(r^2))$ sparse cover to improve the space further from $O(n\cdot\log n\cdot\log_r\Phi)$ to $O(n\cdot r^2\cdot\log_r\Phi)$.\atodo{Could we remove dependence on aspect ratio? I am not sure, this probably just to build SPT in each cluster. How to combine while keeping such small query time?}

\section{Conclusion and Open Problems}
This paper is mainly concerned with sparse covers for minor free graphs. Our first contribution is to transform from the recently introduced notion of buffered cop decomposition \cite{CCLMST24} to sparse covers. As a result, we obtain a strong $\left(O(r),O(r^{2})\right)$-\SPCS, and strong $\left(4+\eps,O(\frac{1}{\eps})^{r}\right)$-\SPCS for $K_r$-minor free graphs (\Cref{thm:MinorFreeCover}). This significantly improves both sparsity and padding parameters compared to previous work, obtains the strong diameter guarantee, and the ``partition based'' property. 
We then use this new covers to construct low distortion and dimension $\ell_\infty$ embedding of $K_r$ minor free graphs (\Cref{cor:EmbeddingMinor}, and \Cref{thm:EmbeddingMinor3}),  improving significantly over the previous work \cite{KLMN04} in both dimension and distortion.
Next, we show several applications of our new sparse covers. Specifically, we obtained better sparse partitions, and as a result get an exponential improvement in the dependence on $r$ for both universal TSP, and metric  universal Steiner tree problems. We also get a similar exponential improvement for the oblivious buy-at-bulk problem. For the name independent routing scheme in $K_r$-minor free graphs, we get quadratic improvement in the stretch and exponential improvement in the space (in the dependence on $r$) compared with previous work \cite{AGMW10} (or just ginormous improvement in the space compared with \cite{BLT14}). We also obtain similar improvements for path reporting distance oracles.

Our work leaves several open questions:
\begin{enumerate}
	\item The main open question is to obtain improve sparse covers for $K_r$-minor free graphs. Note that we don't have any lower bound beyond what follows from general graphs. In particular, it might be possible to obtain $(O(\log r),O(\log r))$-sparse cover scheme, or $(O(1),\poly(r))$-sparse cover scheme. 
	\item Metric embeddings into $\ell_\infty$. In \Cref{thm:EmbeddingMinor3} we got distortion $3+\eps$ which looks like a natural barrier for sparse cover based techniques. It is very interesting to see if we can get bellow $3$, even for planar graphs. Optimally, we would like to see an embedding with distortion $1+\eps$ and $g(r,\eps)\cdot\log n$ coordinated, for an arbitrary function $g$.
	\item Universal Steiner tree. In this paper, using the reduction from \cite{JLNRS05} we obtain a solution for the metric \UST problem in $K_r$-minor free graphs with stretch $O(\frac{r^4}{\log r})\cdot \log n$. However, for the usual (non-metric) \UST problem, the state of the art for $K_r$-minor free (or even planar) graphs is $O(\log^7n)$ \cite{BCFHHR24}, the same as for general graphs. It is interesting to see if we can use the topological structure to get an improved stretch.
\end{enumerate}

\section*{Acknowledgments}
The author would like to thank James R Lee for helpful discussions regarding the embedding of planar graphs into $\ell_\infty$ from \cite{KLMN04}. The author would also like to thank Shmuel Tomi Klein for the reference to \cite{KN76}.

{\small 
	\bibliographystyle{alphaurlinit}
	\bibliography{LSObib}}
\newpage
\appendix

\section{Proof of \Cref{lem:Huffman}}\label{app:Codes}
We restate  \Cref{lem:Huffman} for convenience.
\Hufmann*
\begin{proof}
	Recall Huffman prefix free code \cite{Huffman52}. Huffman constructs a binary tree
	such that each edge is associated with $\pm1$, and the elements $X$
	are the leafs. Then $h(x)$ is simply the string obtain by traversing
	the unique path from the root to $x$. It is clearly prefix free.
	The tree constructed as follows: pick the two elements $x,y\in X$
	with minimal weight. Let $z\notin X$, and define a new set $X'=X\setminus\left\{ x,y\right\} \cup\left\{ z\right\} $
	with the same weight function $\mu$ on $X\setminus\left\{ x,y\right\} $,
	and let $\mu(z)=\mu(x)+\mu(y)$. Note that $\mu(X')=\mu(X)$. Construct
	a binary tree recursively on $X'$. Finally add back $x,y$ as the
	two children of $z$, assigning $+1$ (resp. $-1$) to the edge $\{z,x\}$ (resp. $\{z,y\}$). As a result we obtain a binary tree. 
	Huffman \cite{Huffman52} proved that this code minimizes the average length: $\sum_{x\in X}\mu(x)\cdot|h(x)|$. In contrast, here we are interested in bounding the length of each code word separately.
	
	To prove the lemma, it is enough to argue for every $x$, that the
	root to $x$ unique path has length at most $2\cdot\left\lceil \log\frac{\mu(X)}{\mu(x)}\right\rceil $.
	The proof is by induction on $|X|$. The base case, where $|X|=1$
	clearly holds as the path is of length $0$. Next we assume $|X|\ge2$.
	Consider an arbitrary point $x\in X$. If $\mu(x)\ge\frac{1}{2}\mu(X)$,
	then $x$ will be ``contracted'' exactly once, and thus again the
	lemma holds. Otherwise, if in the entire process $x$ is contracted
	at most twice, than again we are done (as $2\cdot\left\lceil \log\frac{\mu(X)}{\mu(x)}\right\rceil\ge 2$). Thus we can assume that $x$ was contracted at least three times.
	
	Consider the the first time when $x$ is combined with some element
	$y$ into a point $z$. At this time the set was $X_{i}$, and $x,y$
	had the minimal weight. In particular for every $q\in X_{i}\setminus\left\{ x,y\right\} $,
	$\mu(x)\le\mu(q)$. Consider the first time the point $z$ is contracted
	with a point $q$ into $z'$. $q$ is combined from elements in $X_{i}\setminus\left\{ x,y\right\} $,
	and hence $\mu(x)\le\mu(q)$. It follows that $\mu(z')=\mu(y)+\mu(x)+\mu(q)>2\cdot\mu(x)$.
	Note that in the created tree, the path from $x$ to the root goes
	though $z'$. By the induction hypothesis, there will be a path from
	$z'$ to the root of length at most $\left\lceil 2\cdot\log\frac{\mu(X)}{\mu(z')}\right\rceil $.
	It follows that there will be a path from $x$ to the root of length at most
	\[
	2+\left\lceil 2\cdot\log\frac{\mu(X)}{\mu(z')}\right\rceil =\left\lceil 2+2\cdot\log\frac{\mu(X)}{\mu(z')}\right\rceil =\left\lceil 2\cdot\left(1+\log\frac{\mu(X)}{\mu(z')}\right)\right\rceil =\left\lceil 2\cdot\log\frac{2\mu(X)}{\mu(z')}\right\rceil \le\left\lceil 2\cdot\log\frac{\mu(X)}{\mu(x)}\right\rceil \,,
	\]
	as required.
\end{proof}

\end{document}